\documentclass[12pt]{article}
\usepackage{a4wide,named}
\usepackage{breakurl}             

\usepackage{pstricks,amssymb}

\usepackage[pdftex]{graphicx}

\usepackage{pst-tree}
\usepackage{lscape}
\usepackage{caption}
\usepackage{hyperref}

\usepackage[all]{xy}

\newcommand{\truth}{\top}
\newcommand{\falsity}{\bot}

\newcommand{\tick}{\surd}
\newcommand{\cross}{\times}

\newcommand{\prune}{PRUNE}
\newcommand{\prunez}{PRUNE$_0$}

\newcommand{\longversion}[1]{#1}
\newcommand{\shortversion}[1]{}

\newcommand{\atoms}{AP}
\newcommand{\natn}{N}
\newcommand{\clos}{cl}

\newtheorem{theorem}{THEOREM}
\newtheorem{lemma}[theorem]{LEMMA}

\newenvironment{proof}{
\begin{quotation}PROOF:}{
$\square$ \end{quotation}}

\begin{document}

\title{A traditional tree-style tableau for LTL\longversion{:\\ LONG VERSION}}
\author{Mark Reynolds,\\
The University of Western Australia.\\
{\tt mark.reynolds@uwa.edu.au}
\thanks{
The author would like to 
thank staff, students
and especially Prof. Angelo Montanari at Universit\`a degli studi di Udine,
for interesting and useful discussions
on the development of this approach
while the author was kindly hosted on sabbatical there
in 2014.}
}

\maketitle

  \begin{abstract}
Propositional linear time temporal logic (LTL) is the standard temporal logic for computing applications and many reasoning techniques and tools have been developed for it. Tableaux for deciding satisfiability have existed since the 1980s. However, the tableaux for this logic do not look like traditional tree-shaped tableau systems and their processing is often quite complicated. We present a new simple traditional-style tree-shaped tableau for LTL and prove that it is sound and complete. As well as being simple to understand, to introduce to students and to use manually, it also seems simple to implement and promises to be competitive in its automation. It is particularly suitable for parallel implementations.

  \end{abstract}

  \newcommand{\webpage}{\url{http://www.csse.uwa.edu.au/~mark/research/Online/ltlsattab.html}}
  
  Note:
the latest version of this report can be found via
\webpage.

\section{Introduction}
\label{sec:intro}

Propositional linear time temporal logic,
LTL, is important for hardware and software verification\cite{DBLP:conf/spin/RozierV07}.
LTL satisfiability checking (LTLSAT) is receiving renewed interest
with advances computing power, several industry ready tools,
some new theoretical techniques,
studies of the relative merits of different approaches,
implementation competitions, and benchmarking:
\cite{Goranko2010113,VSchuppanLDarmawan-ATVA-2011,DBLP:conf/spin/RozierV07}.
Common techniques
include
automata-based approaches 
\cite{VaW94,RV11}
and
resolution
\cite{DBLP:journals/aicom/LudwigH10}
as well as tableaux \cite{Gou89,Wol85,DBLP:conf/cav/KestenMMP93,Sch98}.
Each type of approach
has its own advantages and disadvantages
and each can be competitive
at the industrial scale
(albeit within the limits of what may
be achieved with PSPACE complexity).
The state of the art in tableau reasoners for
LTL satisfiability testing is the
technique from \cite{Sch98}
which is used in portfolio reasoners such as
\cite{VSchuppanLDarmawan-ATVA-2011}.

 Many LTL tableau approaches
 produce a very untraditional-looking
 graph,
 as opposed to a tree,
 and need the whole graph to be present
 before
 a second phase of discarding takes place.
 Within the stable of tableau-based approaches to
 LTLSAT, the system
 of \cite{Sch98}
 stands out in
 in being tree-shaped (not a more general graph),
 and in being one-pass,
 not relying on a two-phase building and pruning process.
 It also stands out in speed
 \cite{Goranko2010113}.
 However,
 there are still elements
 of communication between separate branches
 and a slightly complicated
 annotation of nodes with depth measures that
 needs to be managed
 as it feeds in to the tableau rules.
So it is not in the traditional 
style of classical tableaux \cite{Mor12}.

This paper presents
a new simpler tableau for LTL.
It builds on ideas from \cite{SiC85}
and is influenced by LTL tableaux
by \cite{Wol85} and \cite{Sch98}.
It also uses some ideas
from a CTL* tableau approach in \cite{Rey:startab} where
uselessly long branches are curtailed.
The general shape of the
tableau and
its construction rules are mostly unsurprising but the two novel
PRUNE rules
are perhaps a surprisingly simple way to curtail
repetitive branch extension and may be applicable
in other contexts.

The tableau search
allows completely independent
searching down separate branches
and so lends itself to parallel
computing.
In fact this approach is
``embarrassingly parallel''
\cite{Fos95}.
Thus there is also potential for quantum
implementations.
Furthermore, only formula set labels need to be recorded
down a branch, and checked back up the one branch,
and so there is great potential for
very fast implementations.

The soundness, completeness and termination of the tableau search is
proved. The proofs are mostly straightforward.
However, the completeness proof with the PRUNE rules
has some interesting reasoning.

We provide a simple demonstration 
prototype reasoning tool to allow readers
to explore the tableau search process.
The tool is not optimised for speed in any way
but
we report on a small range of
experiments on standard benchmarks which demonstrate that the 
new tableau will be competitive
with the current state of the art \cite{Sch98}.

In this paper
we give some context
in Section~\ref{sec:other},
before we briefly confirm our standard version of the
well-known syntax and semantics for LTL
in Section~\ref{sec:synsem},
describe our tableau approach is general terms
section~\ref{sec:tab},
present the rules section~\ref{sec:rules},
make some comments and provide some motivation for our approach
section~\ref{sec:motiv},
prove soundness in section~\ref{sec:sound},
prove completeness in section~\ref{sec:complete},
and briefly discuss
complexity and implementation issues section~\ref{sec:complex},
and detailed comparisons with the Schwendimann approach in Section~\ref{sec:compare},
before a conclusion
section~\ref{sec:concl}.

The latest version of this long report
can be found at
\url{http://www.csse.uwa.edu.au/~mark/research/Online/ltlsattab.html}
with links to a Java implementation of the tool
and
full details of experiments.

\shortversion{
Full versions of the (short) proofs can be found
in an online technical report
\cite{ltlsattablong}.}

\section{Context and Short Summary of Other Approaches}
\label{sec:other}

LTL is an important logic and there has been sustained 
development of techniques and tools
for working with it over more than half a century.
Many similar ideas appear as parts of different
theoretical tools and it is hard for a researcher
to be across all the threads.
Thus it is worth putting the ideas here in 
some sort of context.

\subsection{Satisfiability checking versus Model Checking}

We will define structures and formulas
more carefully
in the next section
but a structure is essentially a way
of definitively and unambiguously describing an infinite 
sequence of states
and an LTL formula
may or may not apply to the sequence:
the sequence may or may not be a model of the formula.
See section 3 or \cite{Pnu77} for details.

We are addressing a computational task called 
{\em satisfiability checking}.
That is, given an LTL formula,
decide whether or not 
there exists any structure at all which is
a model of that formula.
Input is the LTL formula,
output is yes or no.
The procedure or algorithm must terminate
with the correct answer.

This problem is PSPACE complete in the size of formula
\cite{SiC85}.

There is a related but separate task called {\em model checking}.
Model checking is 
the task of working out whether
a given structure is a model or not
of a given formula.
Input is the LTL formula and a description of a system,
output is yes or no,
whether or not there is a
behaviour generated by the system
which is a model of the formula
(or some variant on that).
Model checking may seem to be more computationally demanding
in that there is a formula and system
to process.
On the other hand,
as described in
\cite{DBLP:conf/spin/RozierV07}
and outlined below,
it is possible to use
 model checking algorithms to do satisfiability checking
 efficiently (in PSPACE).
Model-checking itself, with
a given formula
on a given structure,
is also
PSPACE-complete \cite{SiC85},
but in \cite{LiP85} we can find
an algorithm that is exponential in the size of the formula
and linear in the size of the model.

\subsection{What is a tableau approach
as compared to an automata-based approach or 
a reduction to model-checking}

The thorough experimental comparison in
\cite{VSchuppanLDarmawan-ATVA-2011}
finds that across 
a wide range of carefully chosen benchmarks,
none of the usual three
approaches to LTL satisfiability testing
dominates.
Here we briefly 
introduce tableaux and
distinguish the 
automata-based model-checking approach.
Resolution is a very different technique
and we refer the reader elsewhere, for example to
\cite{DBLP:journals/aicom/LudwigH10,FDP01,hustadt2003trp++}.

Tableau approaches trace their 
origins to the semantic tableaux,
or truth trees,
for classical logics as developed
by Beth \cite{Beth55} and Smullyan \cite{Smu68}.
The traditional tree style remains
for many modal logics \cite{Girle00}
but, as we will
outline below,
this is not so for the LTL logic.

The tableau itself is (typically)
a set of nodes,
labelled by single formulas or sets of formulas,
with a successor relation between nodes.
The tree tableau has a root
and the successor relation
gives each node $0$, $1$ or $2$ children.
The tableau is typically depicted with the
root at the top and
the children below their parent:
thus giving an upside tree-shape.

Building a tableau constitutes
a decision procedure for the logic
when it is governed by
rules for labelling children nodes,
for counting a node as the
leaf of a successful branch
and for counting a node
as the leaf of a failed branch.
Finding one successful branch
typically means that the original
formula is satisfiable while
finding all branches failing
means it is not satisfiable.
There are efficiently implemented 
tableau-based
LTL satisfiability reasoning tools,
which are easily available,
such as {\tt pltl} \cite{PLTL}
and
{\tt LWB} \cite{LWB}.  
We describe the approach and its variations in more 
detail in the next subsection.

The other main approach to 
LTL satisfiability checking
is based on a reduction of the
task to a model-checking question,
which itself is often
implemented via automata-based
techniques.
The traditional automata approach to model checking
seen, for example in \cite{VaW94},
is to compose the model
with an automaton that
recognizes the negation of the property of interest
and check for emptiness
of the composed structure.

Recently there has been some serious development
of satisfiability checking 
on top of automata-based model-checking that
ends up having some tableau-like aspects.
The automata approach via model-checking
as described in
\cite{DBLP:conf/spin/RozierV07}
shows that
if you can do model checking you can do satisfiability checking.
One can construct an automaton
which will 
represent a universal system allowing all possible traces
from the propositions.
A given formula is satisfiable
iff
the universal system
contains a model of the formula.
By not actually 
building the universal system
entirely the overall task
can be accomplished in PSPACE in the size of formula.
Some fast and effective 
approaches to satisfiability checking can then be
developed by
reasoning 
about the model checking
using symbolic 
(bounded) SAT or BDD-based reasoners such as
CadenceSMV and NuSMV
\cite{DBLP:conf/spin/RozierV07}.
Alternatively,
{\em explicit} automata-based model-checkers
such as SPIN \cite{Hol97}
could be used for the model checking
but \cite{DBLP:conf/spin/RozierV07}
shows these to not be competitive.

An important part of the symbolic approach is
the construction of a symbolic automaton
from the formula
and this involves a 
 tableau-like construction
 with sets of subformulas determining states.
The symbolic automaton or tableau
presented as part of the model-checking procedure in
 \cite{CGH97}
has been commonly used
but this has been extended in \cite{RV11}
to a portfolio of translators.
Also an LTL to tableau tool used in \cite{RV11} is available via
the first author's website.
We discuss these sorts of tableaux
again briefly below.

The model-checking approach
to LTL satisfiability checking is further
developed 
in \cite{LZPVH13} where
a novel, on-the-fly 
process combines the automata-construction
and emptiness check.
The tool out-performs
previously existing tools for LTLSAT
which use model-checking.
Comprehensive experimental comparisons
with approaches to LTLSAT beyond
automata-theoretic model-checking based ones
is left as future work.

\subsection{
Different shapes of tableaux
and different ways to the search through the tableau}

In some modal logics,
and in temporal logics in particular,
variations on
the traditional tableau idea have been prevalent
and the pure tree-shape
is left behind as a more
complicated graph of nodes is constructed.
As we will see, this
may happen in these variant tableaux,
when the successor relation is allowed to have
{\em up-links}
from descendants to ancestors,
{\em cross-links}
from nodes to nodes on other branches,
or where there is just an arbitrary
directed graph of nodes.

Such a {\em graph-shaped} approach 
results if, for example,
we give a {\em declarative}
definition of a node and the successor relation
determined by the labels on the node
at each end.
Typically,
the node is identified with its label.
Certain sets of formulas are allowed
to exist as nodes and we have at most one node
with a given label.
Examples of this sort include the tableaux in \cite{Gou84} and
(the symbolic automaton of ) \cite{CGH97}.

Alternatively, 
the tableau may have a 
traditional form that is essentially {\em tree-shped} with
a root and branches of nodes descending
and branching out below that.
Usually a limited form of up-link is allowed
back from a node (leaf or otherwise)
to one of its ancestors.
There may be tableaux which contain several different
nodes with the same label.
Examples, include \cite{Schwe98} and the new one.
The new style tableau even allows multiple nodes down the
same branch with the same labels
while this is not permitted in \cite{Schwe98}.

The tableau construction may be 
{\em incremental}, where only reachable states are constructed,
versus {\em declarative},
when we just define what labels are present in the tableau
and which pairs of labels
are joined by a directed edge.

A tableau is said to be {\em one-pass}
if the construction process
only build legitimate nodes as it proceeds.
On the other hand, it is {\em multi-pass}
if there is an initial construction phase
followed by a culling phase in which
some of the nodes (or labels)
which were constructed are removed as
not being legitimate.

A construction of a tableau,
or a search through a tableau,
will often proceed 
in a {\em depth-first} manner
starting at a chosen node and then 
moving successively to successors.
An alternative is via some sort of {\em parallel}
implementation 
in which branches are explored concurrently.
Search algorithms may make use of heuristics
in
guessing a good branch to proceed on.
Undertaking a depth-first search in a graph-shaped
label-determined
tableau may seem to be similar to
building a tree-shaped tableau 
but 
it is likely that the algorithm will
behave differently when 
it visits a label that has been seen before
down an earlier branch.
In a tree-shaped tableau this may not need to be
recorded.

Which is faster?
Trees \cite{Schwe98}
may have
2EXPTIME worst case complexity.
Graphs re-use labels
and make many possible branches 
at one time:
in general EXPTIME.
However, \cite{Goranko2010113}
demonstrated that
the tree-shaped approach of \cite{Schwe98}
(consistently and sometimes drastically) outperformed 
the graph-shaped approach of \cite{Wol85}.

\subsection{The Wolper and Schwendiman Tableau}

Wolper's \cite{Wol83,Wol85}
was the first LTL tableau.
It is a multi-pass, graph-shaped tableau.
The nodes are labelled with sets of formulas 
(from the closure set)
with a minimal amount of extra notation
attached to record which 
formulas already have been decomposed.
One builds the graph starting at $\{ \phi \}$
and using the standard sorts of decomposition
rules and the transition rule.
The tableau may start off looking like a tree
but there must not be repeated labels
so edges generally end up 
heading upwards and/or crossing branches.
After the construction
phase there is iterated elimination
of nodes according to rules
about successors
and eventualities.
(This approach was later extended to
cover the inclusion of past-time operators
in \cite{LP00:igpl}).

There was a similar but slightly quicker
proposal for a multi-pass, graph-shaped
tableau for LTL in \cite{Gou84}.
The similar tableau in \cite{DBLP:conf/cav/KestenMMP93}
is incremental but multi-pass as it
builds a graph from initial states,
then looks for strongly connected components
(to satisfy eventualities).
Other graph-like tableaux include
those in
\cite{SGL97,MaP95}.

A graph-shaped tableau for LTL also forms
part of the model-checking
approach suggested in \cite{CGH97}.
Here the need to 
check fulfilment of eventualities
is handed over to some CTL fairness constraint
checking on a structure
formed from the product
of the  tableau and the model to be checked.
The symbolic model checker SMV
is used to check the property
subject to those fairness constraints.
In \cite{RV11}, the `symbolic automaton' approach based
on the tableau from \cite{CGH97} was adapted
to tackle LTL satisfiability checking.
Figure~\ref{fig:cgh97tab}
shows a typical graph-style tableau
from \cite{CGH97}.

\begin{figure}
\centering
\includegraphics[width=10cm,trim= 4cm 8.5cm 4cm 4.5cm,clip=true]{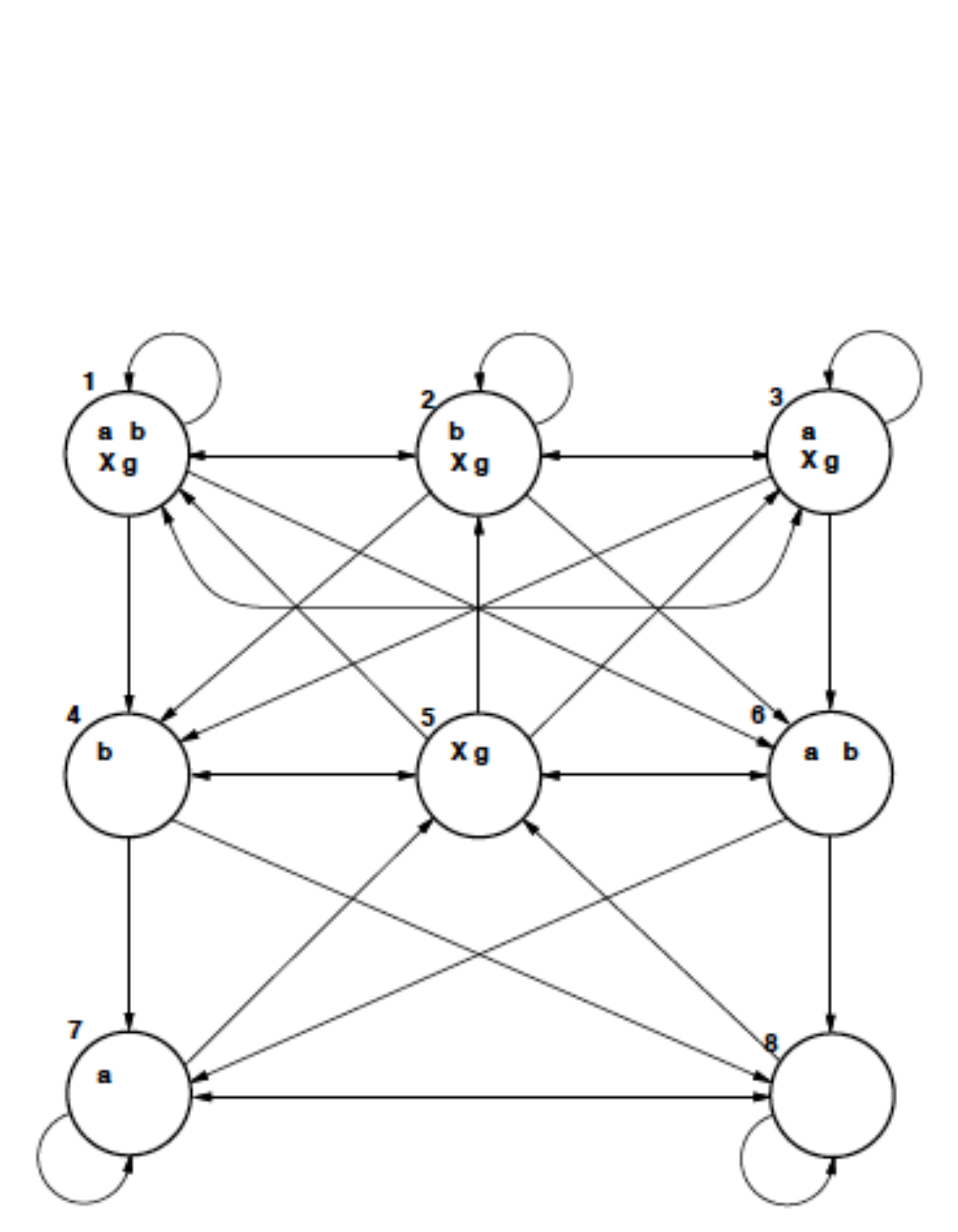}
\caption{A graph-shaped tableau for $g= a U b$ from [CGH97]}
\label{fig:cgh97tab}
\end{figure}

Schwendimann's tableau \cite{Schwe98}
is close to being purely tree-shaped (not a more general graph).
It is also one-pass,
 not relying on a two-phase building and pruning process.
 However,
 there are still elements
 of communication between separate branches
 and a slightly complicated
 annotation of nodes with depth measures that
 needs to be managed
 as it feeds in to the tableau rules.
 In general, to decide on
 success or not for a tableau
 we need to work back up the tree towards the root,
 combining information from
 both branches if there are children.
 Thus, after some construction downwards,
  there is an iterative process
 moving from children to parents
 which may need to wait for both children to
 return their data.
 This can be done by tackling one branch at a time
 or in theory, in parallel.
 The data passed up consists in an index number
 and a set of formulas (unfulfilled eventualities).
 We describe this approach in more detail
 in Section~\ref{sec:compare}.
Figure~\ref{fig:schwetab}
shows a typical tree-style tableau
from \cite{Schwe98}.

\begin{figure}
\centering
\includegraphics[width=10cm,trim= 4cm 11cm 4cm 11cm,clip=true]{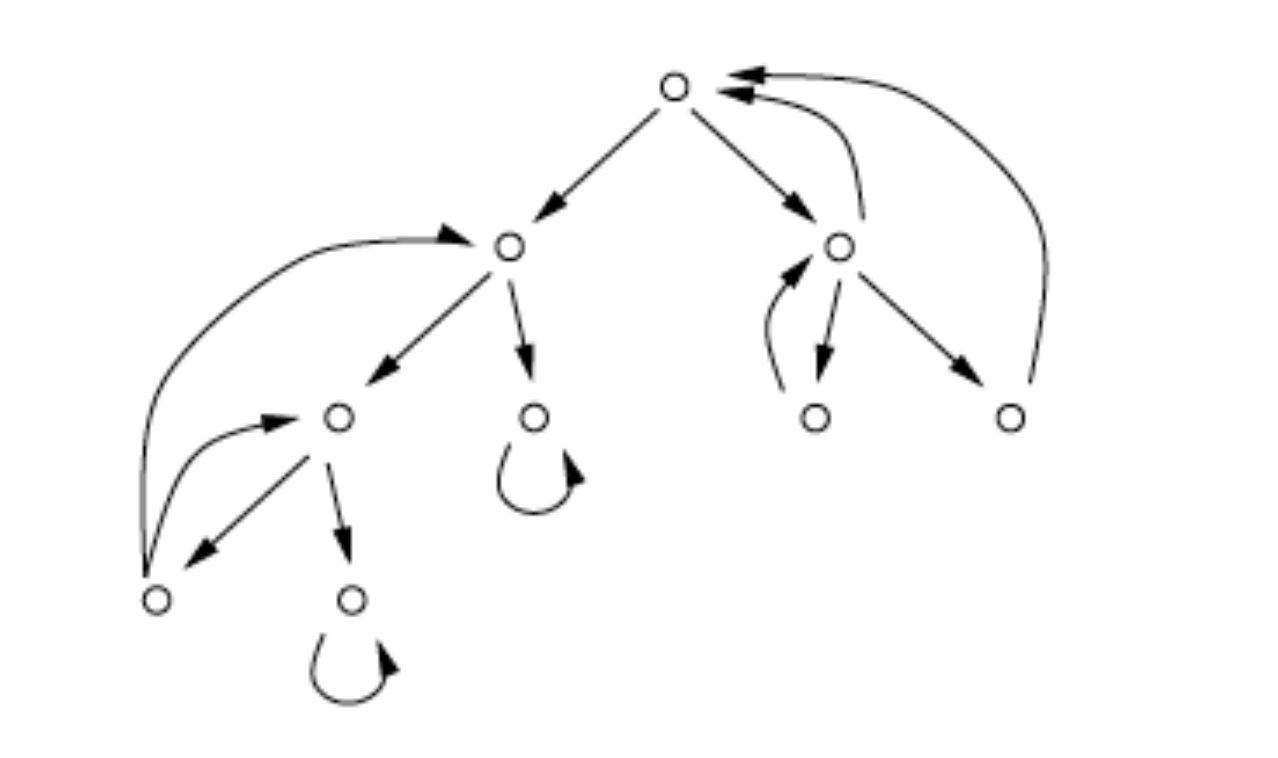}
\caption{A tree-shaped tableau from [Sch98]}
\label{fig:schwetab}
\end{figure}

\subsection{Implementations}

The main available
tools for
LTL satisfiability
checking
are listed and described
 in \cite{VSchuppanLDarmawan-ATVA-2011}.
 These included
 {\tt pltl}
 which (along with another tableau option)
 implements Schwendiammn's
 approach.

We should also mention \cite{LZPVH13} here.
The tool uses a novel, on-the-fly approach to LTLSAT
via model-checking
and out performs
previously existing tools for LTLSAT
which use the model-checking approach.
There is an open source LTL to tableau translator available from Rozier's website
used for the \cite{RV11} approach.
Also, a more limited implementation was released open-source following the publication of \cite{CGH97}: it is called ltl2smv and distributed with the NuSMV model checker.

\subsection{Benchmarking}

There is a brief comparison of 
the tableau approaches of
Schwendimann and Wolper 
 in \cite{Goranko2010113} and 
in \cite{RV11} there is a comparison
symbolic model-checking approaches.

However,
the most thorough benchmarking exercise
is as follows.
Most known implemented tools for deciding satisfiability
in LTL are compared in \cite{VSchuppanLDarmawan-ATVA-2011}.
The best tools from three classes are chosen and compared:
automata-based reduction to model-checking,
tableau and resolution.
A large range of benchmark patterns are collected or newly proposed.
They find that no solver,
no one of the three approaches,
dominates the others.
The tableau tool ``pltl''
based on Schwendimann's approach
is the best of the tableaux
and the best overall on
various classes of pattern.
A portfolio solver is suggested and also evaluated.

The benchmark formulas,
rendered in a selection of different formats,
is available from Schuppan's webpage.

\subsection{So what is novel here?}

The two PRUNE rules are novel.
They force construction of branches to be terminated in certain circumstances.
They depend only on the labels at the node and the labels of its ancestors.
In general they may allow a label to be repeated some number of times
before the termination condition is met.

The overall tableau shape is novel.
Although tableaux are traditionally tree-shaped,
no other tableau system for LTL 
builds graphs that are tree shaped.
Most tableaux for LTL are more complicated graphs.
The Schwendimann approach is close to being tree-shaped
but there are still up-links from non-leaves.

The labels on the tableau are just
sets of formulas from the closure set
of the original formula 
(that is, subformulas and a few others).
Other approaches
(such as Schwendimann's)
require other annotations on nodes.

Overall tableau: one that is in completely traditional style (labels are sets
of formulas), tree shaped tableau construction,
no extraneous recording of calculated values, just looking at the labels.

Completely parallel development of branches.
No communication between different branches.
This promises interesting and useful
parallel implementations.

The reasoning speed seems to be 
(capable of being)
uniquely fast on some important benchmarks.

\section{Syntax and Semantics}
\label{sec:synsem}

We assume a countable set $\atoms$
of propositional atoms,
or atomic propositions.

A (transition) structure
is a triple $(S,R,g)$
with
$S$ a finite set of states,
$R \subseteq S \times S$
a binary relation
called the transition relation
and labelling $g$ tells us which atoms are true
at each state: for each $s \in S$,
$g(s) \subseteq \atoms$.
Furthermore, $R$ is assumed to be total: every state has at least one successor
$\forall x \in S.  \exists y \in S \mbox{ s.t.} (x,y) \in R$

Given a structure $(S,R,g)$
an $\omega$-sequence of states
$\langle s_0, s_1, s_2, ... \rangle$
from $S$ is a {\em fullpath}
(through $(S,R,g)$)
iff
for each $i$,
$(s_i,s_{i+1}) \in R$.
If $\sigma= \langle s_0, s_1, s_2, ... \rangle$
is a fullpath
then
we write
$\sigma_i=s_i$,
$\sigma_{\geq j}= \langle s_j, s_{j+1}, s_{j+2}, ... \rangle$
(also a fullpath).

The (well formed) formulas of LTL
include the atoms and
if $\alpha$ and $\beta$ are formulas then so are
$\neg \alpha$, $\alpha \wedge \beta$,
$X \alpha$, and $\alpha U \beta$.

We will also include some formulas
built using other connectives
that are often presented as abbreviations instead.
However, before detailing them
we present the semantic clauses.

Semantics defines
truth of formulas
on a fullpath through a structure.
Write $M, \sigma \models \alpha$ iff the
formula $\alpha$ is true of the fullpath $\sigma$
in the structure $M=(S,R,g)$ defined recursively by:\\
\begin{tabular}{lll}
$M, \sigma \models p$ & iff & $p \in g(\sigma_0)$, for $p \in {\atoms}$;\\
$M, \sigma \models \neg \alpha$ & iff &
$M, \sigma \not \models \alpha$;\\
$M, \sigma \models \alpha \wedge \beta$ & iff &
$M, \sigma \models \alpha$ and 
$M, \sigma \models \beta$;\\
$M, \sigma \models X \alpha$ & iff &
$M, \sigma_{\geq 1} \models \alpha$; and\\
$M, \sigma \models \alpha U \beta$ & iff &   
there is some $i \geq 0$ s.t.
$M, \sigma_{\geq i} \models \beta$ and\\&&
for each $j$, if $0 \leq j < i$ then
$M, \sigma_{\geq j}  \models \alpha$.\\ 
\end{tabular}

Standard Abbreviations in LTL
include the
classical
$\truth \equiv p \vee \neg p$,
$\falsity \equiv \neg \truth$,
$\alpha \vee \beta \equiv \neg ( \neg \alpha \wedge \neg \beta)$,
$\alpha \rightarrow \beta \equiv \neg \alpha \vee \beta$,
$\alpha \leftrightarrow \beta \equiv
( \alpha \rightarrow \beta ) \wedge ( \beta \rightarrow \alpha)$.
We also have the temporal:
$F \alpha \equiv (\truth U \alpha)$,
$G \alpha \equiv \neg F ( \neg \alpha)$
read as eventually and  always respectively.
In reasoning with
LTL, it is
simpler to remove these abbreviations
from input formulas
and then deal with a relatively small set
tableau rules
for the disabbreviated language.
However, 
experience with
tableaux and typical real life LTL
examples 
gives a strong indication that
automated reasoning is quicker
if these abbreviations
are included as 
first-class language constructs in their own rights.
Thus, inputs are accepted in the
larger language including these
symbols,
they are not disabbreviated
and 
there are enough tableau
rules to
process the bigger set of formulas
directly.
In this paper
we present a tableau system
with the
larger set of rules.

A formula $\alpha$ is {\em satisfiable} iff there
is some structure $(S,R,g)$
with some fullpath
$\sigma$ through it such that
$(S,R,g), \sigma \models \alpha$.
A formula is {\em valid} iff 
for all structures $(S,R,g)$
for all fullpaths
$\sigma$ through $(S,R,g)$ we have
$(S,R,g), \sigma \models \alpha$.
A formula is valid iff its negation 
is not satisfiable.

For example,
$\truth$, $p$, $Fp$, $p \wedge Xp \wedge F \neg p$,
$Gp$
are each satisfiable.
However,
$\falsity$, $p \wedge \neg p$,
$Fp \wedge G \neg p$,
$p \wedge G(p \rightarrow Xp) \wedge F \neg p$
are each not satisfiable.

We will fix a particular formula,
$\phi$ say,
and describe
how a tableau for $\phi$ is built
and
how that decides the
satisfiability
or otherwise, of $\phi$.
We will use other formula names
such as $\alpha$, $\beta$, e.t.c.,
to indicate
arbitrary formulas
which are used in labels
in the tableau for $\phi$.

\section{General Idea of the Tableau}
\label{sec:tab}

The tableau for $\phi$ is a tree of nodes
(going down the page from a root)
each labelled by 
a set of formulas.
To lighten the notation,
when we present a tableau in a diagram
we will omit the braces $\{ \}$
around the sets which form labels.
The root is labelled $\{ \phi\}$.

Each node has 0, 1 or 2 children.
A node is called a leaf if it has 0 children.
A leaf may be
crossed ($\times$),
indicating a failed branch,
or ticked ($\surd$),
indicating a successful branch.
Otherwise, a leaf indicates an unfinished branch
and unfinished tableau.

The whole tableau is successful if there is a ticked branch.
This indicates a ``yes'' answer to the satisfiability of $\phi$.
It is failed if all branches are crossed: indicating ``no''.
Otherwise it is unfinished.
Note that you can stop the algorithm,
and report success if you tick a 
branch even if other branches have not reached
a tick or cross yet.

A small set of tableau rules (see below) determine
whether a node has one or two children
or whether to cross or tick it.
This depends on the label of the parent,
and also, for some rules, on labels
on ancestor nodes,
higher up the branch towards the root.
The rule also determines
the labels on the children.

\newcommand{\horizontalTransition}{$\rightarrow \hspace{-0.45cm} / \hspace{-0.1cm} / \hspace{0.2cm}$}
\newcommand{\verticalTransition}{$\downarrow \hspace{-0.25cm} =$}

The parent-child relation
is indicated by a vertical arrow in diagrams (if needed).
However, to indicate use of one particular rule
(coming up below)
called the TRANSITION rule
we will 
use a vertical arrow (\verticalTransition)
with two strikes across it,
or just an equals sign.

A node label may be the empty set, although
it then can be immediately ticked by rule EMPTY below.

A formula which is
an atomic proposition,
a negated atomic proposition
or of the form
$X \alpha$ or $\neg X \alpha$
is called {\em elementary}.
If a node label is non-empty
and there are no direct contradictions,
that is no $\alpha$ and $\neg \alpha$
amongst the formulas in the label,
and every formula it contains is elementary then
we call the label (or the node) {\em poised}.

\newcommand{\cupdot}{\mathbin{\mathaccent\cdot\cup}}

Most of the rules {\em consume} formulas.
That is,
the parent may have a label
$\Gamma = \Delta \cupdot \{ \alpha \}$,
where
$\cupdot$ is disjoint union,
and a child may have a label
$\Delta \cup \{ \gamma \}$ so that
$\alpha$ has been removed,
or consumed.

See the
$\neg p \wedge X \neg p \wedge (q Up)$
example given in
Figure~\ref{fig:eg9}
of a simple but succesful tableau.

\begin{figure}

\centering
\includegraphics[width=10cm,trim= 5cm 15cm 5cm 4.5cm,clip=true]{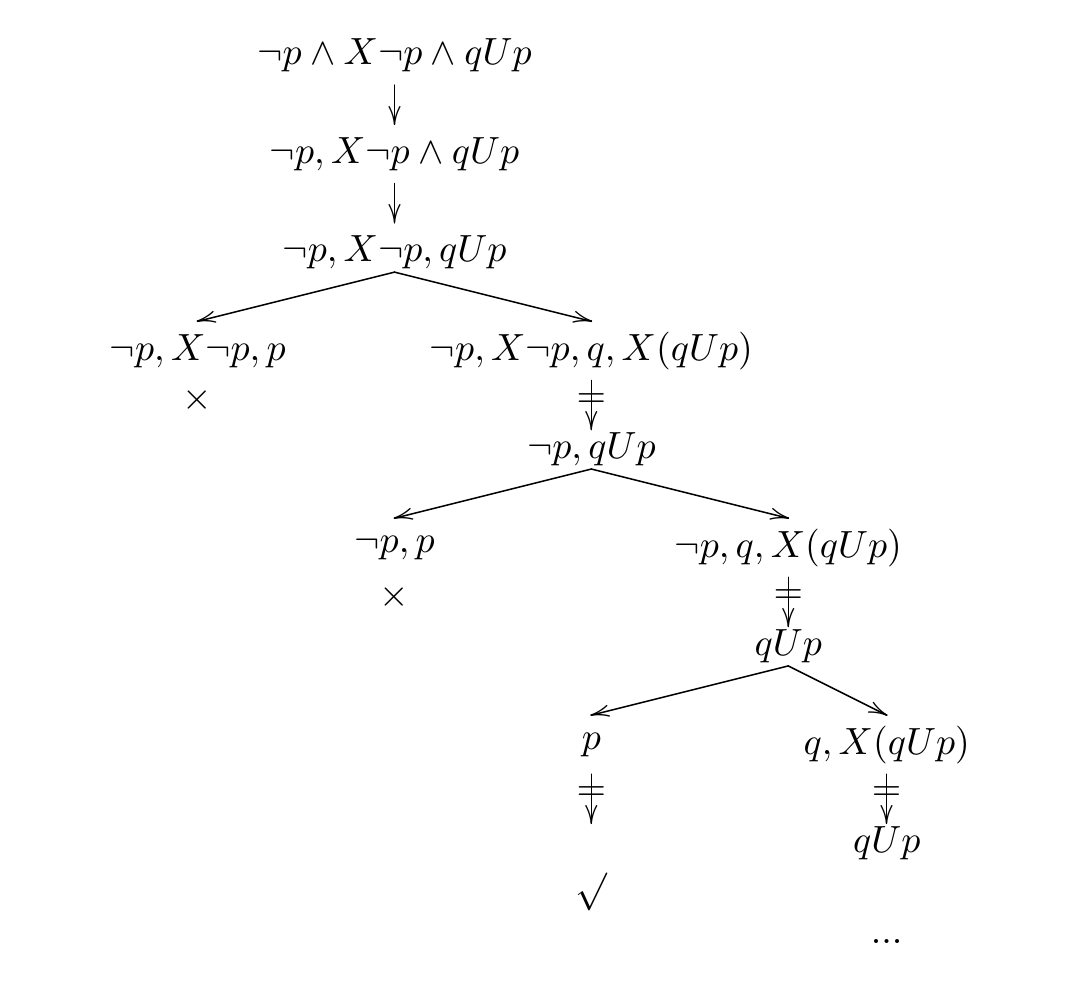}

\caption{$\neg p \wedge X \neg p \wedge (q Up)$}
\label{fig:eg9}
\end{figure}

As usual a tableau node $x$ is an ancestor of a node $y$
precisely when $x=y$ or $x$ is a parent of $y$
or a parent of a parent, etc.
Then $y$ is a descendent of $x$
and we write $x \leq y$.
Node $x$ is a proper ancestor of $y$,
written $x < y$, iff
it is an ancestor and $x \neq y$.
Similarly proper descendent.
When we say that a node $y$ is between
node $x$ and its descendent $z$,
$x \leq y \leq z$, then we mean that
$x$ is an ancestor of $y$ and
$y$ is an ancestor of $z$.

A formula of the form $X(\alpha U \beta)$ 
or $X F \beta$ appearing in a poised
label of a node $m$, say, also plays an important role.
We will call such a formula
an {\em $X$-eventuality} because
$\alpha U \beta$ or $F \beta$ is
often called an {\em eventuality},
and its truth depends on $\beta$ being eventually
true in the future (if not present).
If the formula $\beta$ appears 
in the label of a proper descendent node $n$ of $m$
then we say that the $X$-eventuality
at $m$
has been
{\em fulfilled} by $n$ by $\beta$ being satisfied there.
 
\section{Rules:}
\label{sec:rules}

There are twenty-five rules altogether. We would only need ten for the 
minimal LTL-language, but recall that
we are treating the usual abbreviations as
first-class symbols, so they each need
a pair of rules.

Most of the rules are what we call {\em static} rules.
They tell us about 
formulas that may be true at a single state in a model.
They 
determine the number,
$0$, $1$ or $2$, of child nodes
and the labels on those nodes
from the label on the current parent node
without reference to any other labels.
These rules are unsurprising
to anyone familiar with any of the 
previous LTL tableau approaches.

To save repetition of wording we use
an abbreviated notation for presenting
each rule: the rule $A / B$
relates the parent label $A$
to the child labels $B$.
The parent label
is a set of formulas.
The child labels
are given as either
a $\tick$ representing the leaf of a successful branch, 
a $\cross$ representing the leaf of a failed branch, a 
single set being the label on the single child
or a pair of sets
separated by a vertical bar $|$ being the respective
labels on a pair of child nodes.

Thus, for example, the $U$-rule, means
that if node is labelled 
$\{ \alpha U \beta \} \cupdot \Delta$,
if we choose to use the
$U$-rule and
if we choose to decompose $\alpha U \beta$ using the rule
then
the node will have two children labelled
$\Delta  \cup \{  \beta\}$
and 
$\Delta  \cup \{ \alpha,  X(\alpha U \beta)\}$
respectively.

Often, several different rules may be applicable to a
node with a certain label.
If another applicable rule is chosen,
or another formula is chosen to be decomposed
by the same rule,
then the child labels may be different.
We discuss this non-determinism later.

These are the positive static rules:\\
\begin{tabular}{ll}
{EMPTY-rule:} &
$\{  \}  \;/\; \tick$.\\
{$\truth$-rule:} &
$\{ \truth \} \cupdot \Delta \;/\; \Delta$.\\
{$\falsity$-rule:} &
$\{ \falsity \} \cupdot \Delta \;/\; \times$.\\
{$\wedge$-rule:} &
$\{ \alpha \wedge \beta \} \cupdot \Delta  \;/\;  
( \Delta  \cup \{  \alpha, \beta \} )$.\\
{$\vee$-rule:} &
$\{ \alpha \vee \beta \} \cupdot \Delta  \;/\;  
( \Delta  \cup \{  \alpha \}
\; |\; \Delta  \cup \{ \beta \} )$.\\
{$\rightarrow$-rule:} &
$\{ \alpha \rightarrow \beta \} \cupdot \Delta  \;/\;  
( \Delta  \cup \{  \neg \alpha \} 
\; |\;  \Delta  \cup \{ \beta \})$.\\
{$\leftrightarrow$-rule:} &
$\{ \alpha \leftrightarrow \beta \} \cupdot \Delta  \;/\;  
( ( \Delta  \cup \{  \alpha, \beta \} \; |\; 
\Delta  \cup \{ \neg \alpha, \neg \beta \} )) $.\\
{$U$-rule:} &
$\{  \alpha U \beta  \} \cupdot \Delta  \;/\;  
( \Delta  \cup \{  \beta \} \; |\; 
\Delta  \cup \{ \alpha, X( \alpha U \beta) \}) $.\\
{$F$-rule:} &
$\{ F \alpha \} \cupdot \Delta  \;/\;  
( \Delta  \cup \{  \alpha \} \; |\; 
\Delta  \cup \{ XF \alpha \}) $.\\
{$G$-rule:} &
$\{ G\alpha \} \cupdot \Delta  \;/\;  
\Delta  \cup \{ \alpha, XG \alpha\}$.\\
\end{tabular}

There are also static rules
for negations:\\
\begin{tabular}{ll}
{CONTRADICTION-rule:} &
$\{ \alpha, \neg \alpha \} \cupdot \Delta \;/\; \times$.\\
{$\neg \neg$-rule:} &
$\{ \neg \neg \alpha \} \cupdot \Delta \;/\; \Delta \cup \{ \alpha \}$.\\
{$\neg \falsity$-rule:} &
$\{ \neg \falsity \} \cupdot \Delta \;/\; \Delta$.\\
{$\neg \truth$-rule:} &
$\{ \neg \truth \} \cupdot \Delta \;/\; \times$.\\
{$\neg \wedge$-rule:} &
$\{ \neg (\alpha \wedge \beta) \} \cupdot \Delta  \;/\;  
( \Delta  \cup \{  \neg \alpha \}
\; |\; \Delta  \cup \{ \neg \beta \} )$.\\
{$\neg \vee$-rule:} &
$\{ \neg (\alpha \vee \beta) \} \cupdot \Delta  \;/\;  
( \Delta  \cup \{  \neg \alpha, \neg \beta \} )$.\\
{$\neg \rightarrow$-rule:} &
$\{ \neg (\alpha \rightarrow \beta) \} \cupdot \Delta  \;/\;  
( \Delta  \cup \{  \alpha, \neg \beta \})$.\\
{$\neg \leftrightarrow$-rule:} &
$\{ \neg ( \alpha \leftrightarrow \beta) \} \cupdot \Delta  \;/\;  
( ( \Delta  \cup \{  \alpha, \neg \beta \} \; |\; 
\Delta  \cup \{ \neg \alpha, \beta \} )) $.\\
{$\neg U$-rule:} &
$\{  \neg ( \alpha U \beta)  \} \cupdot \Delta  \;/\;  
( \Delta  \cup \{  \neg \alpha, \neg \beta \} \; |\; 
\Delta  \cup \{ \neg \beta, X \neg ( \alpha U \beta) \}) $.\\
{$\neg G$-rule:} &
$\{ \neg G \alpha \} \cupdot \Delta  \;/\;  
( \Delta  \cup \{  \neg \alpha \} \; |\; 
\Delta  \cup \{ X \neg G \alpha \}) $.\\
{$\neg F$-rule:} &
$\{ \neg F \alpha \} \cupdot \Delta  \;/\;  
\Delta  \cup \{ \neg \alpha, X \neg F \alpha\}$.\\
\end{tabular}

The remaining four non-static rules are only applicable
when a label is poised
(which implies that none of the static rules will
be applicable to it).
In presenting them we 
use the convention that
a node $u$ has label $\Gamma_u$.
The rules are to be considered 
in the following order.

{\bf [LOOP]:}
If a node $v$ 
with poised label $\Gamma_v$
has a 
proper ancestor
(i.e. not itself)
$u$ with poised label $\Gamma_u$ such that
$\Gamma_u \supseteq \Gamma_v$,
and for each $X$-eventuality
$X(\alpha U \beta)$ or $XF \beta$ in $\Gamma_u$
we have 
 a node $w$
such that $u < w \leq v$
and $\beta \in \Gamma_w$
then
$v$ can be a ticked leaf.

{\bf [\prune]:}
Suppose that
$u < v < w$
and each of
$u$, $v$ and $w$
have the same poised label 
$\Gamma$.
Suppose also that
for each $X$-eventuality $X(\alpha U \beta)$
or $XF \beta$ in $ \Gamma$,
if there is $x$ with
$\beta \in \Gamma_x$
and $v < x \leq w$
then
there is $y$ such that
$\beta \in \Gamma_y$
and $u < y \leq v$.
Then $w$ can be a crossed leaf.

{\bf [\prunez]:}
Suppose that
$u < v$
share the same poised label $\Gamma$.
Suppose also that 
$\Gamma$ contains at least one $X$-eventuality
but there is 
no $X$-eventuality $X(\alpha U \beta)$ or $XF \beta$ in $\Gamma$,
with a node $x$
such that
is $\beta \in \Gamma_x$
and $u < x \leq v$.
Then $v$ can be a crossed leaf.

{\bf [TRANSITION]:}
If none of the other rules above do apply to it
then
a node labelled by 
poised $\Gamma$ say,
can have one child
whose label is: 
$\Delta = \{ \alpha |
X \alpha \in \Gamma \}
\cup
\{ \neg \alpha |
\neg X \alpha \in \Gamma \}
$.

\section{Comments and Motivation}
\label{sec:motiv}

A traditional classical logic style tableau
starts with the formula
in question and breaks it down into
simpler formulas
as we move down the page.
The simpler formulas
being satisfied should ensure that the
more complicated parent label
is satisfied.
Alternatives are presented as branches.
See the
example given in
Figure~\ref{fig:eg4}.

\begin{figure}
\centering
\begin{minipage}{.5\textwidth}
  \centering
\includegraphics[width=8cm,trim= 5cm 20cm 5cm 4.5cm,clip=true]{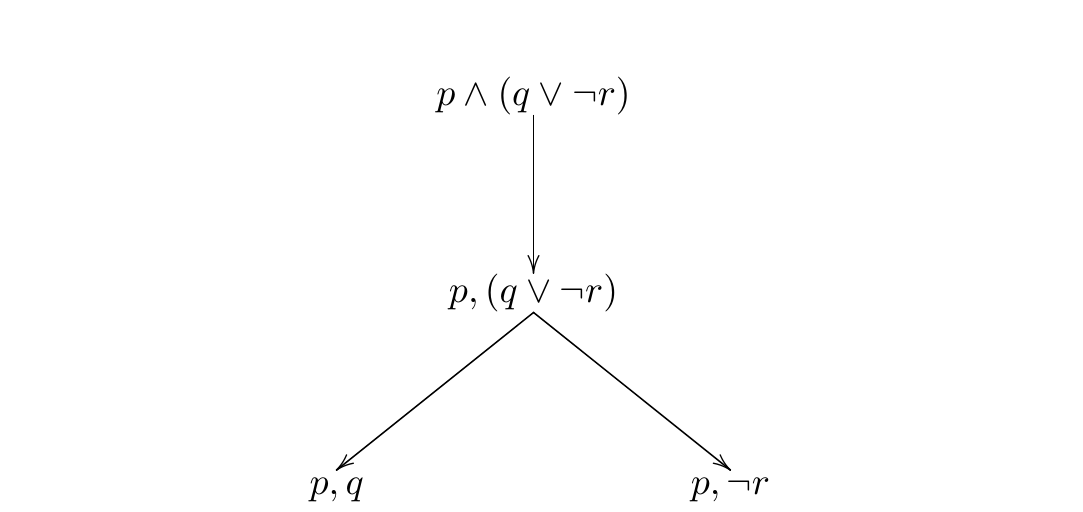}
  \captionof{figure}{Classical disjunction}
  \label{fig:eg4}
\end{minipage}%
\begin{minipage}{.5\textwidth}
  \centering
\includegraphics[width=8cm,trim= 5cm 19cm 5cm 4.5cm,clip=true]{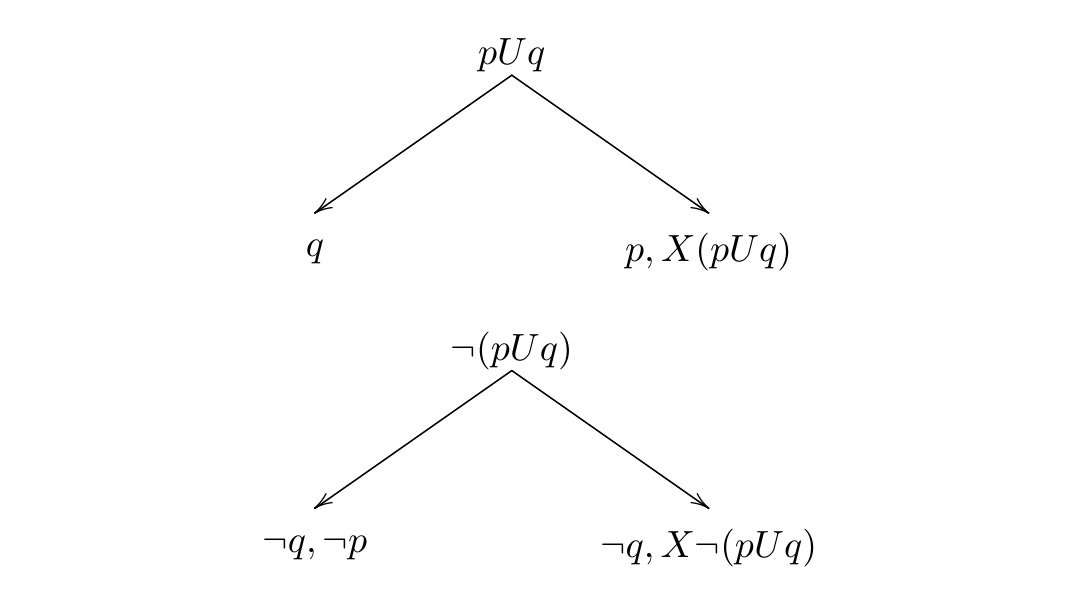}
  \captionof{figure}{Until also gives us choices}
  \label{fig:eg6}
\end{minipage}
\end{figure}

We follow this style of tableau
as is evident by the
classical look of the tableau rules
involving classical connectives.
The $U$ and $\neg U$ rules
are also in this vein,
noting that temporal formulas
such as 
$U$ also gives us choices:
Figure~\ref{fig:eg6}.

Eventually, we break down a formula
into elementary ones.
The atoms and their negations
can be satisfied immediately
provided there are no contradictions,
but to reason about the $X$ formulas
we need to move forwards in time.
How do we do this?
\longversion{
See Figure~\ref{fig:eg7}.

\begin{figure}
\centering
\begin{minipage}{.5\textwidth}
  \centering
\includegraphics[width=8cm,trim= 5cm 18cm 5cm 4.5cm,clip=true]{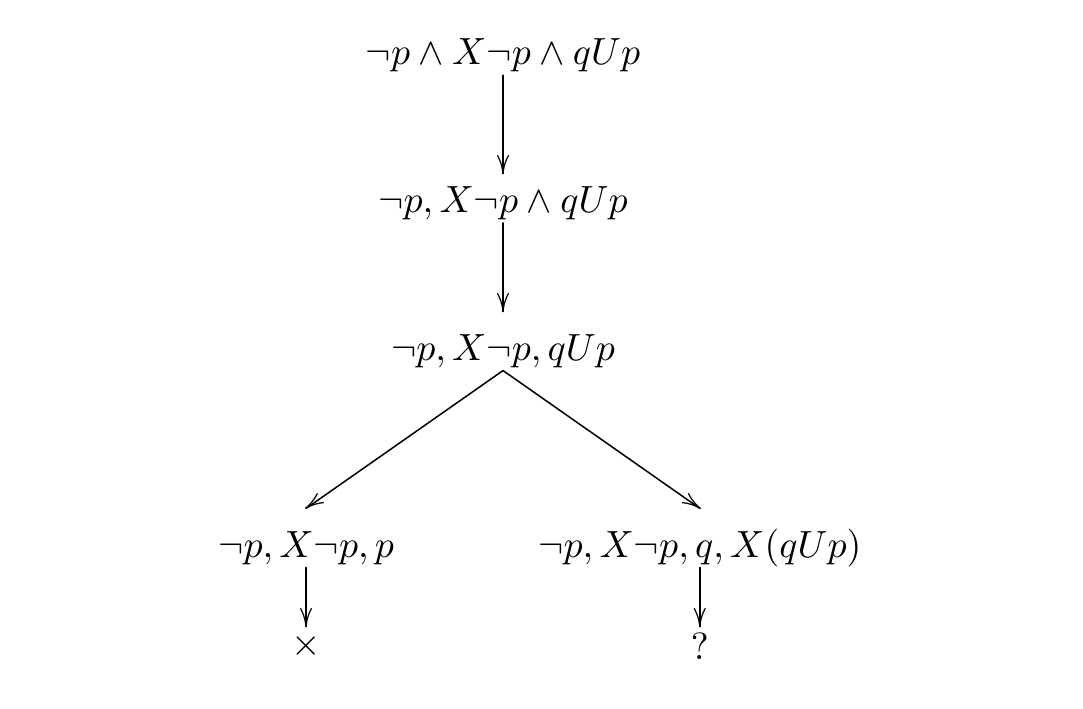}
  \captionof{figure}{But what to do when we want to move forwards in time?}
  \label{fig:eg7}
\end{minipage}%
\begin{minipage}{.5\textwidth}
  \centering
\includegraphics[width=8cm,trim= 5cm 18cm 5cm 4.5cm,clip=true]{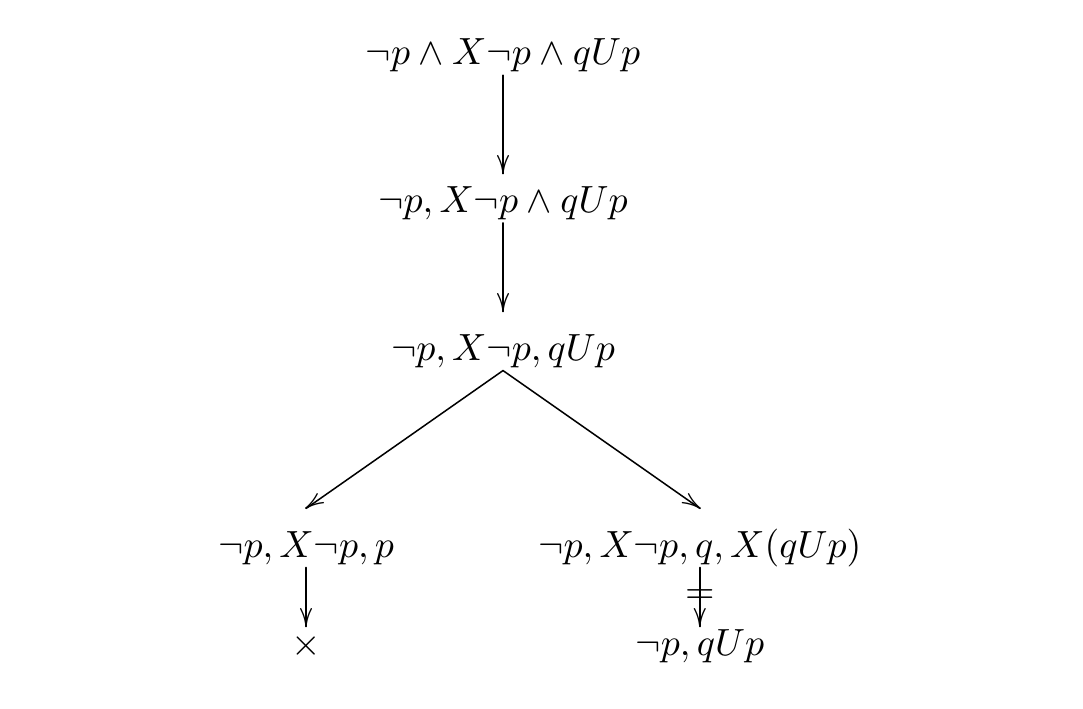}
  \captionof{figure}{Introduce a new type of TRANSITION}
  \label{fig:eg8}
\end{minipage}
\end{figure}

The answer is that we introduce a new type of TRANSITION step:
see
Figure~\ref{fig:eg8}.
Reasoning switches to the next time point
and we carry over only information
nested below $X$ and $\neg X$.

}

\shortversion{
The answer is that we introduce a new type of TRANSITION step
at this stage:
see
Figure~\ref{fig:eg9}.
Reasoning switches to the next time point
and we carry over only information
nested below $X$ and $\neg X$.

}

 With just these rules
we can now do the whole
$\neg p \wedge X \neg p \wedge (q Up)$
example. See
Figure~\ref{fig:eg9}.

This example is rather simple, though,
and we need additional rules
to deal with infinite behaviour.
 Consider the example
$Gp$
which,
in the absence of additional rules, gives rise to a very repetitive infinite tableau.
Figure~\ref{fig:eg10}.
\begin{figure}
\centering
\begin{minipage}{.5\textwidth}
  \centering
\includegraphics[width=10cm,trim= 7cm 18cm 5cm 4.5cm,clip=true]{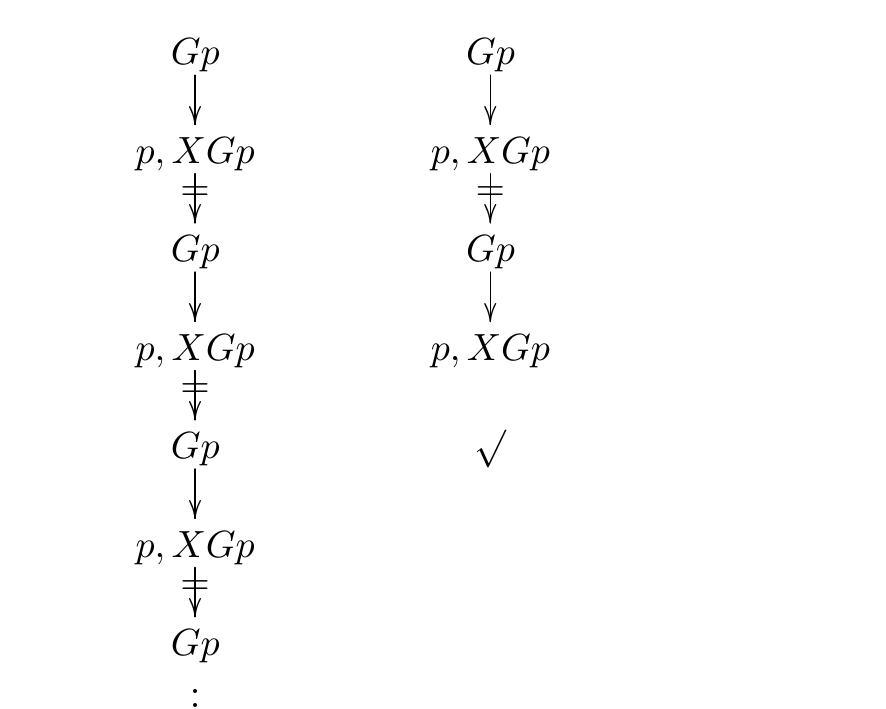}
  \captionof{figure}{$Gp$ gives rise to a very repetitive infinite tableau without the LOOP rule,
  but succeeds quickly with it}
  \label{fig:eg10}
\end{minipage}%
\begin{minipage}{.5\textwidth}
  \centering
\includegraphics[width=10cm,trim= 5cm 16cm 1cm 4.5cm,clip=true]{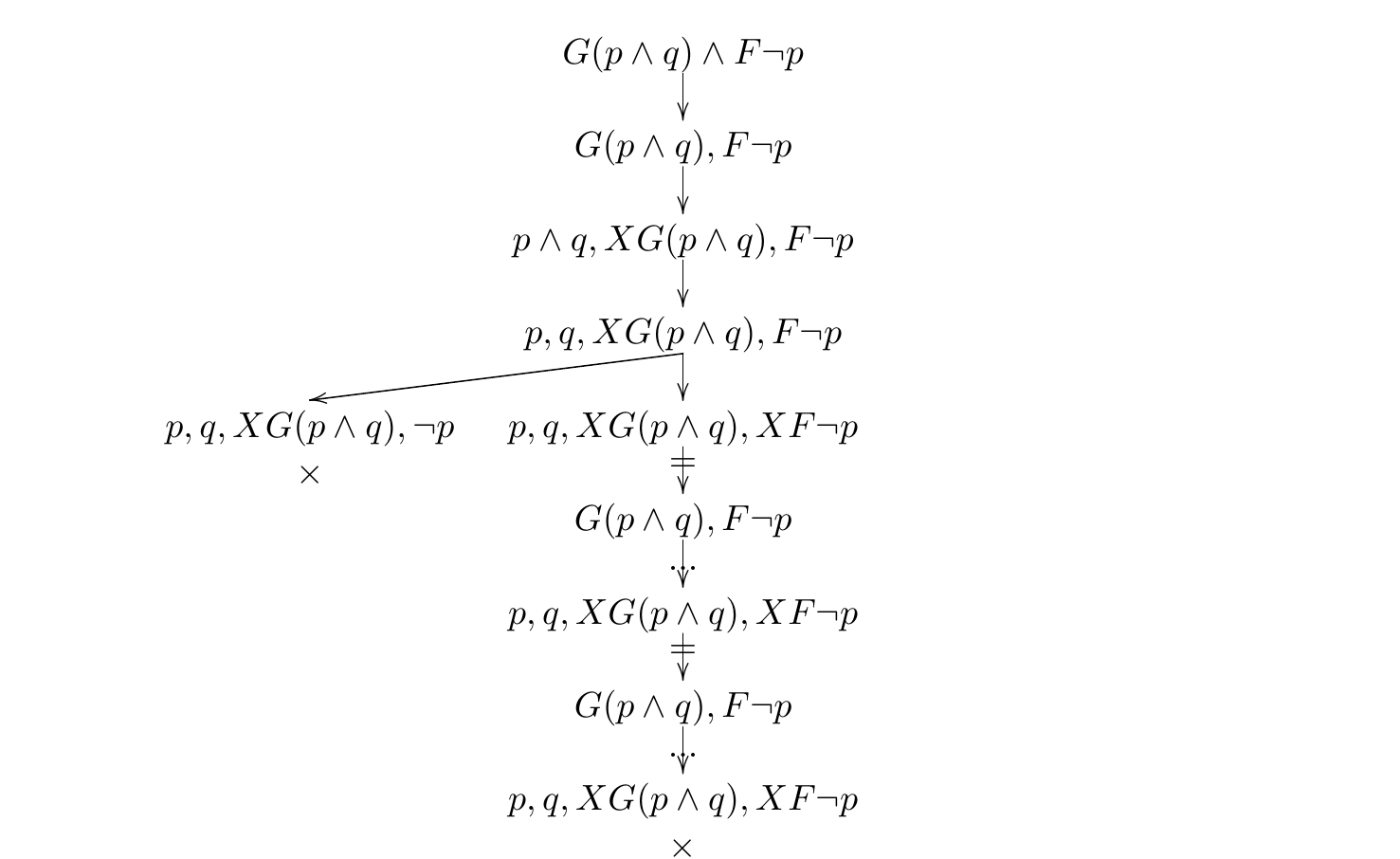}
  \captionof{figure}{$G(p \wedge q) \wedge F \neg p$ crossed by PRUNE}
  \label{fig:eg11}
\end{minipage}
\end{figure}
Notice that the infinite fullpath that it suggests is a model for $Gp$
as would be a fullpath just consisting of the one state
with a self-loop (a transition from itself to itself).

This suggests that we should allow
the tableau branch construction to halt if a
state is repeated.
However the
example $G(p \wedge q) \wedge F \neg p$ shows
that we can not just accept infinite loops
as demonstrating satisfiability:
the tableau for this
unsatisfiable formula would have an 
infinite branch if we did
not use the PRUNE rule
to cross it
(Figure~\ref{fig:eg11}).
Note that the more specialised \prunez\ rule
can be used to cross the branch
one TRANSITION earlier.

Notice that the infinite fullpath that 
the tableau suggests is this time
not a model for $G(p \wedge q) \wedge F \neg p$.
Constant repeating of
$p, q$ being made true does not
satisfy the conjunct
$F \neg p$.
We have postponed the {\em eventuality}
forever
and this is not acceptable.

If $\alpha U \beta$ appears in 
the tableau label of a node $u$
then we want
$\beta$ to appear in the label of 
some later (or equal node) $v$.
In that case we say that the eventuality
is {\em satisfied}
by $v$.

Eventualities are eventually satisfied
in any (actual) model of a formula:
by the semantics of $U$.

Thus we
introduce the LOOP
rule with an
extra condition.
If a label is repeated along a branch
and all eventualities
are satisfied in between
then
we can build a model
by looping states.
In fact, the ancestor can have a superset
and it will work (see the soundness proof below).

Examples like $G(p \wedge q) \wedge F \neg p$
(in Figure~\ref{fig:eg11})
and
$p \wedge G(p \rightarrow Xp)
\wedge F \neg p$
which have branches that go on forever
without satisfying eventualities,
still present a problem for us.
We need to stop and fail branches
so that we can answer ``no'' correctly and terminate
and so that we do not get distracted 
when another branch may be successful.
In fact, no infinite branches should be allowed.

The final rule that we
consider,
and the most novel,
is based on
observing that
these infinite branches
are just getting repetitive without
making a model.
The repetition
is clear because
there are only a finite set of 
formulas
which can ever appear in labels
for a given initial formula $\phi$.
\longversion{
The {\em closure set} for a formula $\phi$
is as follows:
\[
\{
\psi, \neg \psi | \psi \leq \phi \}
\cup 
\{ X (\alpha U \beta),
\neg X (\alpha U \beta) |
\alpha U \beta \leq \phi 
\}
\]
Here we use $\psi \leq \phi$ to mean that
$\psi$ is a subformula
of $\phi$.
The size of closure set is
$\leq 4n$
where $n$ is the length of 
the initial formula.
Only formulas
from this finite set
will appear in labels.
So there are only $\leq 2^{4n}$
possible labels.
}

A similar observation
in the case of the branching time temporal logic CTL*
suggested the idea
of {\em useless}
intervals on branches
in the tableau in \cite{Rey:startab}.
It is also related to the 
proof of the small model theorem
for LTL in \cite{SiC85}.

The PRUNE rule is as follows.
If a node at the end of a branch
(of a partially complete tableau)
has a label which has 
appeared already twice above,
and
between the second and third
appearance 
there are no new eventualities satisfied
then 
that whole interval of states
has been useless.
\longversion{The \prunez\ rule
applies similar reasoning to an initial
repeat in which no eventualities are fulfilled.
}
 
It should
be mentioned that the
tableau building process
we describe 
above
is non-deterministic in several respects
and so
really not a proper description of an algorithm.
However, we will see in the
correctness proof below
that the further details 
of which formula to
consider at each step
in building the 
tableau are unimportant.

Finally a suggestion for a nice
example to try.
Try
$p \wedge G(p \leftrightarrow X \neg p)
\wedge G(q \rightarrow \neg p)
\wedge G(r \rightarrow \neg p)
\wedge G(q \rightarrow \neg r)
\wedge GFq
\wedge GFr$.

\section{Proof of Correctness: Soundness:}
\label{sec:sound}

Justification will consist of three
parts:
each established
whether or not the optional rules
are used.

Proof of soundness.
If a formula
has a successful tableau then
it has a model.

Proof of completeness:
If a formula has a model then
building a tableau
will be successful.

Proof of termination.
Show that the 
tableau building
algorithm will always terminate.

First, a sketch of the Proof of Termination.
Any reasonable tableau search
algorithm will always terminate 
because
there can be no infinitely long branches.
We know this because the
LOOP and \prune\ rule will
tick or
cross any that go on too long.
Thus there will either
be at least one tick or all crosses.
Termination is also why we require that
static rules consume formulas
in between 
TRANSITION rules.


Now soundness.
We use a successful tableau to make a model
of the formula,
thus showing that it is satisfiable.
In fact we just use a successful branch.
Each TRANSITION
as we go down the branch
tells us that we are moving from one state
to the next.
Within a particular state we can make all the
formulas listed true there
(as evaluated along the rest of the fullpath).
Atomic propositions listed tell us that they
are true at that state.
An induction deals with most of the rest of the
formulas.
Eventualities either get satisfied and disappear
in a terminating branch
or have to be satisfied 
if the branch is ticked by the LOOP rule.

Suppose that $T$ is a successful tableau for $\phi$.
Say that the branch
$b = \langle x_0, x_1, x_2, ..., x_n \rangle$ of nodes
of $T$ ends in a tick.
Denote by $\Gamma(u)$, the
tableau label on a 
node $u$.
We build $(S,R,g)$ 
from $b$ and its tableau labels.

In fact, there are only a few
$x_i$ that really matter:
each time when we are about to use TRANSITION
and when we are about to use EMPTY or LOOP
to finish (at $x_n$).
Let $j_0, j_1, j_2, ..., j_{k-1}$ be the  indices of nodes
from $b$ 
at which the TRANSITION rule is used.
That is, the TRANSITION rule is used 
to get from $x_{j_i}$ to $x_{j_{i}+1}$.
See
Figure~\ref{fig:eg13}.

\begin{figure}
\centering
\includegraphics[width=8cm]{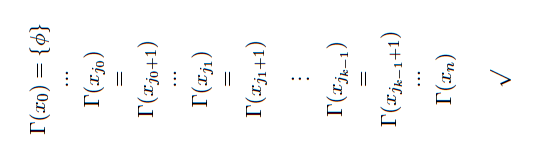}

\caption{TRANSITIONS in $b$: view sideways}
\label{fig:eg13}
\end{figure}

If $b$ ends in a tick from EMPTY
then
let $S= \{ 0, 1, 2, ..., k \}$:
so it contains $k+1$ states.
It is convenient to consider
that $j_k=n$ in the EMPTY case
and put
$\Gamma(x_{j_k}) = \{ \}$.
The states will correspond to
$x_{j_0}, x_{j_1}, ..., x_{j_{k-1}}, x_{j_k}$.

See
Figure~\ref{fig:eg15}.

\begin{figure}
\centering
\includegraphics[width=10cm,trim= 5cm 22.5cm 5cm 4.5cm,clip=true]{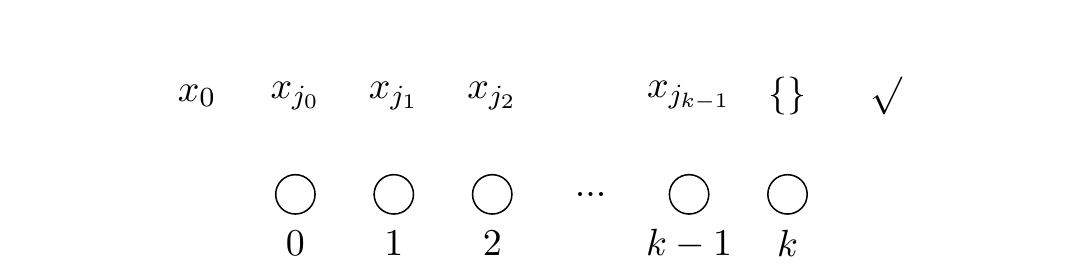}

\caption{$b$ ends in a tick from EMPTY}
\label{fig:eg15}
\end{figure}

If $b$ ends in a tick from LOOP
then
let $S= \{ 0, 1, 2, ..., k-1 \}$:
so it contains $k$ states.
These will correspond to
$x_{j_0}, x_{j_1}, ..., x_{j_{k-1}}$.

See
Figure~\ref{fig:eg16}.

\begin{figure}
\centering
\includegraphics[width=8cm,trim= 4.5cm 21.5cm 4.5cm 4.5cm,clip=true]{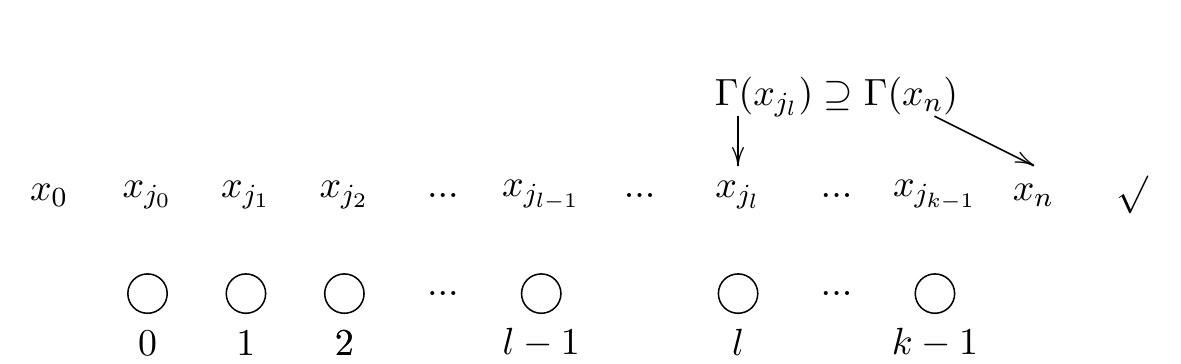}

\caption{$b$ ends in a tick from LOOP}
\label{fig:eg16}
\end{figure}

Let $R$ contain each $(i,i+1)$ for $i<k-1$.
We will also add extra pairs to $R$ to
make a fullpath.

If $b$ ends in a tick from EMPTY
then 
just put 
$(k-1,k)$ and a self-loop $(k,k)$
in $R$ as well.
See
Figure~\ref{fig:eg14}.

\begin{figure}
\centering
\includegraphics[width=10cm,trim= 5cm 22cm 5cm 4.5cm,clip=true]{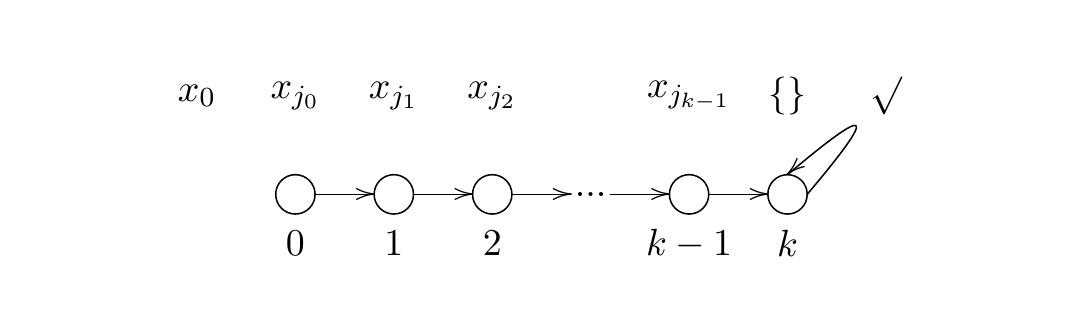}

\caption{$b$ ends in a tick from EMPTY}
\label{fig:eg14}
\end{figure}

If $b$ ends in a tick from LOOP
then 
just put 
$(k-1,l)$ 
in $R$ as well
where $l$ is as follows.
Say that $x_{m}$ is the state that
``matches'' $x_n$.
So look at the application of the LOOP rule that
ended $b$ in a tick.
There is a proper ancestor $x_{m}$ of $x_n$
in the tableau with
$\Gamma(x_{m}) \supseteq \Gamma(x_n)$
and all eventualities
in $\Gamma(x_{m})$ are cured
between $x_{m}$ and $x_n$.
The rule requires
$x_{m}$ to be poised
so it is just before
 a TRANSITION rule.
 So say that $m=j_l$.
Put $(k-1,l) \in R$.
See
Figure~\ref{fig:eg17}.

\begin{figure}
\centering
\includegraphics[width=8cm,trim= 4.5cm 21cm 4.5cm 5cm,clip=true]{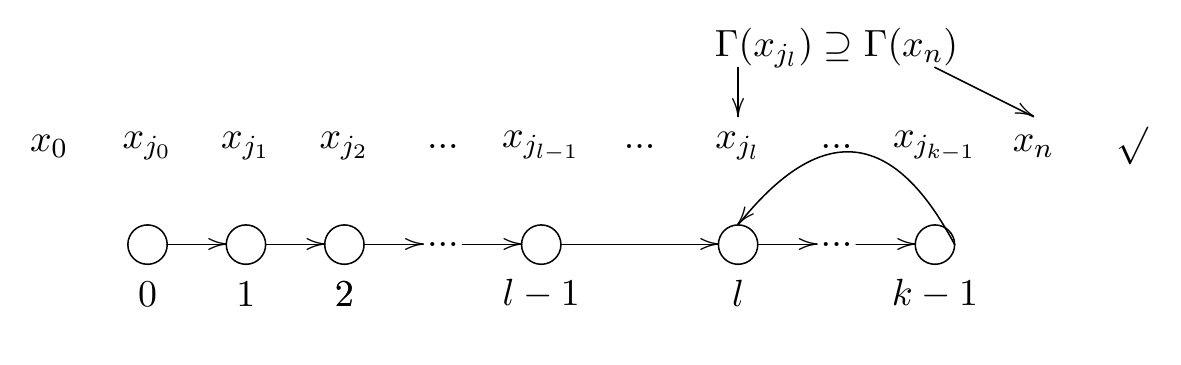}

\caption{$m=j_l$}
\label{fig:eg17}
\end{figure}

A model with a 
line and one loop back is sometimes
called a {\em lasso}
\cite{SiC85}.

Now let us define the labelling $g$ of states by atoms
in $(S,R,g)$.
Let $g(i)= \{ p \in \atoms | p \in \Gamma(x_{j_i}) \}$.

Finally our proposed model of $\phi$ is
along the only fullpath $\sigma$
of $(S,R,g)$ that starts at $0$.
That is,
if $b$ ends in a tick from EMPTY
then 
$\sigma = \langle  0, 1, 2, ..., k-1, k, k ,k ,k , ... \rangle$
while
if $b$ ends in a tick from LOOP
then 
$\sigma = \langle  0, 1, 2, ..., k-2, k-1, l, l+1, l+2, ..., k-2, k-1, l, l+1 , ... \rangle$.

Let $N$ be the length of the first (non-repeating) part of the model:
in the EMPTY case
$N=k-1$ and in the
LOOP case
$N=l$.
Let $M$ be the length of the repeating part:
in the EMPTY case
$M=1$ and in the
LOOP case
$N=k-l$.
So in either case the model has
$N+M$ states $\{ 0, 1, ..., N+M-1\}$
with state $N$ coming (again) after state $N+M-1$ etc.
In particular,
$\sigma_i=i$ for $i<N$
and
$\sigma_i=(i-N)\bmod M +N$
otherwise.

Now we are going to define a set $\Delta_i$ of formulas
for each $i=0, 1, 2, ....$,
that we will want to be satisfied at $\sigma_i$.
They collect formulas in labels in between 
TRANSITIONS, and
loop on forever.
Thus $\Delta_0$ is to be
a set of formulas that 
we want to be true at the first state of the model
and
$\Delta_N$
those true when the model starts to repeat.
In the very special case of
$N=0$,
when $0$ is the first repeating state,
let
$\Delta_0 = \bigcup_{s \leq j_0} \Gamma(x_{s})
\cup
\bigcup_{j_{k-1} < s \leq n} \Gamma(x_{s})$.
If $N>0$ then
let the collection for the first repeating state be
$\Delta_N
= \bigcup_{j_{N-1} < s \leq j_N} \Gamma(x_{s})
\cup
\bigcup_{j_{k-1} < s \leq n} \Gamma(x_{s})$.
So in either case,
$\Delta_N$ has formulas
from two separate sections
of the tableau.

If $N>0$, for $i=0$,
put
$\Delta_0= \bigcup_{s \leq j_0} \Gamma(x_{s})$.
For each $i=1, 2, ..., N+M-1$,
except $i=N$,
let 
$\Delta_i= \bigcup_{j_{i-1} < s \leq j_j} \Gamma(x_{s})$,
all the formulas that
appear between the $(i-1)th$
and $i$th consecutive uses
of the TRANSITION rule.

Finally,
for all $i \geq N+M$,
put
$\Delta_i= \Delta_{i-M}
= \Delta_{(i-N) \bmod M + N}$.

\begin{lemma}
\label{lem:hl2}
If $X \alpha \in \Delta_i$ 
for some $i$,
then 
$\alpha \in \Delta_{i+1}$.
Also,
if $\neg X \alpha \in \Delta_i$ 
for some $i$,
then 
$\neg \alpha \in \Delta_{i+1}$.
\end{lemma}

\begin{proof}
Just consider the $X \alpha$ case:
the $\neg X \alpha$ case is similar.

Choose some $i$ such that
$X \alpha \in \Delta_i$.
As the $\Delta_i$s repeat, we may as well assume
$0 \leq i \leq N+M-1$.
There are three main cases:
$i=N$, $i=N+M-1$ or otherwise.

Consider the $i \neq N$ and $i \neq N+M-1$ case first.
Thus $X \alpha$ appears in 
some 
$\Gamma(x_s)$ for
$j_{i-1} < s \leq j_i$.
Because no static rules 
remove them,
any formula
of the form $X \alpha$
will survive in the tableau labels
until the poised label
$\Gamma(x_{j_i})$
just before a TRANSITION rule is used.
After the TRANSITION rule we will have
$\alpha \in 
\Gamma(x_{j_i+1})$
and $\alpha$ will 
be collected in 
$\Delta_{i+1}$, the next state label collection.

However,
in the other case
when $i=N+M-1$,
we have
$X \alpha$ 
surviving to be in the poised label
$\Gamma(x_{j_{N+M-1}})$.
In that case,
$\alpha$ will be 
in $\Gamma(x_{j_{N+M-1}+1})$ and
collected in
$\Delta_{N}$.
But,
when $i=N+M-1$
then $i+1 = N+M$
and so
$\Delta_{i+1}=\Delta_{N+M}
= \Delta_N \ni \alpha$ as required.

Finally,
in the case
when $i=N$,
we may have
$X \alpha$ 
in 
some
$\Gamma(x_s)$ for
$j_{N-1} < s \leq j_N$
or in
some 
$\Gamma(x_s)$ for
$j_{k-1} < s \leq j_n$.
In the first subcase,
it survives until 
$\Gamma(x_N)$
and the reasoning proceeds as above.
In the second subcase,
$X \alpha \in \Gamma(x_s)$ for
some $s$ with $j_{k-1} < s \leq j_n$.
Thus it survives to
be in 
$\Gamma(x_{j_n})$
when we are about to use the LOOP rule.
However, it will then also be in
$\Gamma(x_{j_l}) \supseteq \Gamma(x_n)$.
After the TRANSITION rule
at $x_{j_l}=x_{j_N}$,
$\alpha$ will be in 
$\Gamma(x_{j_N+1})$ and will
be collected in
$\Delta_{N+1}$ as required.
\end{proof}

\begin{lemma}
\label{lem:hl3}
Suppose $\alpha U \beta \in \Delta_i$.
Then there is some $d \geq i$
such that
$\beta \in \Delta_d$
and for all
$f$,
if $i \leq f < d$
then 
$\{ \alpha, \alpha U \beta, X(\alpha U \beta) \}
\subseteq 
\Delta_f$.
\end{lemma}

\begin{proof}
For all $i$,
whenever $\alpha U \beta \in \Delta_i$
then either $\beta$ will also be there
or both
$\alpha$ and $X (\alpha U \beta)$ will be.
To see this, consider which static rules can be used to remove
$\alpha U \beta$. Only 
the $U-$ and $F$-rule can do this.

By Lemma~\ref{lem:hl2},
if $X ( \alpha U \beta) \in \Delta_i$
then 
$\alpha U \beta \in \Delta_{i+1}$.
Thus, by a simple induction,
$\alpha U \beta$,
and so also the other two formulas,
will be in all
$\Delta_f$ for $f \geq i$
unless $f \geq d \geq i$ with
$\beta \in \Delta_d$.

It remains to show that 
$\beta$ 
does appear in some
$\Delta_d$.

If the branch ended with EMPTY,
then we know this must happen
as the $\Gamma(x_n)$
is empty and so does not contain $X(\alpha U \beta)$.
So suppose
that the branch
ended with
a LOOP
up to tableau node $x_{j_l}$
but that
$\alpha U \beta \in \Delta_f$
for all $f \geq i$.

For some $f >i$, we have
$(f-N) \bmod M =0$,
so we know
$\alpha U \beta \in \Delta_f
= \Gamma(x_{j_l})$.
Thus 
$\alpha U \beta$ is one of the eventualities
in $\Gamma(x_{j_l})$
that have to be satisfied
between
$x_{j_l}$ and $x_n$.

Say that $\beta \in \Gamma(x_h)$
and it will also be in the
next pre-TRANSITION label
$x_{j_q}$ after $x_h$.
So eventually we find a 
$d \geq i$ such that
$(d-N) \bmod M+N=q$
and
$\beta \in \Delta_d$
as required.
\end{proof}

\begin{lemma}
\label{lem:hl4}
Suppose $\neg (\alpha U \beta) \in \Delta_i$.
Then either 1) or 2) hold.
1) There is some $d \geq i$
such that
$\neg \alpha, \neg \beta \in \Delta_d$
and for all
$f$,
if $i \leq f < d$
then 
$\{ \neg \beta, \neg(\alpha U \beta), X\neg (\alpha U \beta) \}
\subseteq 
\Delta_f$.
2) For all $d \geq i$,
$\{ \neg \beta, \neg(\alpha U \beta), X\neg (\alpha U \beta) \}
\subseteq 
\Delta_d$.

\end{lemma}

\begin{proof}
This is similar to Lemma~\ref{lem:hl3}.
\end{proof}

Now we need to show that
$(S,R,g), \sigma \models \phi$.
To do so we prove 
a stronger result:
a {\em truth lemma}.

\begin{lemma}[truth lemma]
for all $\alpha$,
for all $i\geq 0$,
if
$\alpha \in \Delta_i$
then
$(S,R,g), \sigma_{\geq i} \models \alpha$.
\end{lemma}

\begin{proof}
This is proved by induction on the construction
of $\alpha$.
However, we do cases for $\alpha$ and
$\neg \alpha$ together and prove that
for all $\alpha$,
for all $i \geq 0$:
if
$\alpha \in \Delta_i$
then
$(S,R,g), \sigma_{\geq i} \models \alpha$;
and
if
$\neg \alpha \in \Delta_i$
then
$(S,R,g), \sigma_{\geq i} \models \neg \alpha$.

The case by case reasoning is straightforward
given the preceeding lemmas.
See the long version for details.

{Case $p$:}
Fix $i \geq 0$.
If $i<N$ let $i'=i$
and otherwise
let $i'=(i-N)\bmod M +N$.
Thus $\sigma_i=i'$.
If
$p \in \Delta_i= \Gamma(x_{j_{i'}})$
then, by definition of $g$,
$p \in g(i')$.
So
$p \in g(\sigma_i)$
and
$(S,R,g), \sigma_{\geq i} \models p$
as required.
If 
$\neg p \in \Delta_i$
then
(by rule CONTRADICTION)
we did not put $p$ in
$g(i')$ and thus
$(S,R,g), \sigma_{\geq i} \models \neg p$.
This follows as no static rules
remove atoms or negated atoms from
labels.

{Case $\neg \neg \alpha$:}
Fix $i \geq 0$.
If 
$\neg \neg \alpha \in \Delta_i$
then
$(S,R,g), \sigma_{\geq i} \models \neg \neg \alpha$
because
$\alpha$ will also have been put in
$\Delta_i$
(by the $\neg \neg$-rule)
and
so by induction
$(S,R,g), \sigma_{\geq i} \models \alpha$.
Note that the $\falsity$-rule also
removes a double negation
from a label set but it immediately
crosses the branch so it is not relevant here.
$\neg \neg \neg \alpha$ is similar.

{Case $\alpha \wedge \beta$:}
Fix $i \geq 0$.
Suppose
$\alpha \wedge \beta \in \Delta_i$.
We know this formula is removed before the next 
TRANSITION (or the end of the branch if that is sooner).
There are two ways for such a conjunction to be removed:
the $\wedge$-rule,
or if the optional $\leftrightarrow$-rule is able to applied
and is used.
In the first case
$(S,R,g), \sigma_{\geq i} \models \alpha \wedge \beta$
because
$\alpha$ and $\beta$ will also have been put in
$\Delta_i$
(by the $\wedge$-rule)
and
so by induction
$(S,R,g), \sigma_{\geq i} \models \alpha$ and
$(S,R,g), \sigma_{\geq i} \models \beta$.

Suppose instead that
$\alpha = \alpha_1 \rightarrow \beta_1$
and
$\beta = \beta_1 \rightarrow \alpha_1$
and the $\leftrightarrow$-rule is used
to remove $\alpha \wedge \beta$.
Then on branch $b$
either
$\alpha_1$ and $\beta_1$ are included and so
in $\Delta_i$ as well,
or their negations are.
In the first case,
by induction we have
$(S,R,g), \sigma_{\geq i} \models \alpha_1$ and
$(S,R,g), \sigma_{\geq i} \models \beta_1$,
and so we also have
$(S,R,g), \sigma_{\geq i} \models \alpha_1 \rightarrow \beta_1$
and
$(S,R,g), \sigma_{\geq i} \models \beta_1 \rightarrow \alpha_1$
as required.
The second negated case is similar.

Suppose
$\neg (\alpha \wedge \beta) \in \Delta_i$.
Again we know this formula is removed 
and we see that
there are four rules that could cause that to happen:
$\neg \wedge$-rule,
$\truth$-rule,
$\vee$-rule
and
$\rightarrow$-rule.

If
$\neg (\alpha \wedge \beta) \in \Delta_i$
is removed by $\neg \wedge$-rule then
$(S,R,g), \sigma_{\geq i} \models \neg (\alpha \wedge \beta)$
because
we will have 
put
$\neg \alpha \in \Delta_i$
or
$\neg \beta \in \Delta_i$
(or one or both of them are already there)
and so by induction
$(S,R,g), \sigma_{\geq i} \models \neg \alpha$
or
$(S,R,g), \sigma_{\geq i} \models \neg \beta$.

$\truth$-rule:
If
$\truth= \neg (\neg p \wedge \neg \neg p) \in \Delta_i$
is removed by $\truth$-rule then
$(S,R,g), \sigma_{\geq i} \models \neg (\neg p \wedge \neg \neg p)$
anyway so we are done.

$\vee$-rule:
If
$\alpha \vee \beta= \neg (\neg \alpha \wedge \neg \beta) \in \Delta_i$
is removed by $\vee$-rule then
$(S,R,g), \sigma_{\geq i} \models \neg (\neg \alpha \wedge \neg \beta)$
because
we will have 
put
$\alpha \in \Delta_i$
or
$\beta \in \Delta_i$
(or one or both of them are already there)
and so by induction
$(S,R,g), \sigma_{\geq i} \models \alpha$
or
$(S,R,g), \sigma_{\geq i} \models \beta$.

$\rightarrow$-rule:
If
$\alpha \rightarrow \beta= \neg (\neg \neg \alpha \wedge \neg \beta) \in \Delta_i$
is removed by $\rightarrow$-rule then
$(S,R,g), \sigma_{\geq i} \models \neg (\neg \neg \alpha \wedge \neg \beta)$
because
we will have 
put
$\neg \alpha \in \Delta_i$
or
$\beta \in \Delta_i$
(or one or both of them are already there)
and so by induction
$(S,R,g), \sigma_{\geq i} \models \neg \alpha$
or
$(S,R,g), \sigma_{\geq i} \models \beta$.

{Case $\alpha U \beta$:}
If
$\alpha U \beta \in \Delta_i$
then
by the $U$-rule, or the optional $F$-rule,
we will have either
put both
$\alpha \in \Delta_i$
and
$X (\alpha U \beta) \in \Delta_i$
or
we will have
$\beta \in \Delta_i$.

Consider the second case.
$(S,R,g), \sigma_{\geq i} \models \beta$ so
$(S,R,g), \sigma_{\geq i} \models \alpha U \beta$
and we are done.

Now consider the first case:
$\alpha U \beta \in \Delta_i$
as well as
$\alpha \in \Delta_i$
and
$X (\alpha U \beta) \in \Delta_i$.
By Lemma~\ref{lem:hl3},
this keeps being true for later $i' \geq i$
until
$\beta \in \Delta_{i'}$.
By induction, for each $i' \geq i$ until then,
$(S,R,g), \sigma_{\geq {i'}} \models \alpha$.
Clearly if we get to a $l > i$ with
$\beta \in \Delta_l$
then 
$(S,R,g), \sigma_{\geq {l}} \models \beta$
and
$(S,R,g), \sigma_{\geq {i}} \models \alpha U \beta$
as required.

If
$\neg (\alpha U \beta) \in \Delta_i$
then $\neg U$-rule
and $G$-rule
mean that
$\neg \beta, \neg \alpha \in \Delta_i$
or
$\neg \beta, X\neg (\alpha U \beta) \in \Delta_i$.

In the first case,
$(S,R,g), \sigma_{\geq i} \models \neg \alpha $
and
$(S,R,g), \sigma_{\geq i} \models \neg \beta$
so
$(S,R,g), \sigma_{\geq i} \models \neg (\alpha U \beta)$
as required.

In the second case we 
can use Lemma~\ref{lem:hl4}
which uses 
an induction 
to show that
$\neg \beta$, $\neg (\alpha U \beta)$, 
$X \neg (\alpha U \beta)$
keep appearing in the $\Delta_{i'}$ labels
forever
or until $\neg \alpha$ also appears.

In either case
$(S,R,g), \sigma_{\geq i} \models \neg (\alpha U \beta)$
as required.

{Case $X \alpha$:}
If
$X \alpha \in \Delta_i$
then, by Lemma~\ref{lem:hl2}
$\alpha \in \Delta_{i+1}$
 so
 by induction
 $(S,R,g), \sigma_{\geq i+1} \models \alpha$
 and
$(S,R,g), \sigma_{\geq i} \models X \alpha$
as required.

$\neg X \alpha$ is similar.
If
$\neg X \alpha \in \Delta_i$
then, by Lemma~\ref{lem:hl2},
$\neg \alpha \in \Delta_{i+1}$
 so
 by induction (because we did $\neg \alpha$ first)
 $(S,R,g), \sigma_{\geq i+1} \models \neg \alpha$
 and
$(S,R,g), \sigma_{\geq i} \models \neg X \alpha$
as required.

And thus ends the soundness proof.
\end{proof}

If we have a successful tableau then the formula
is satisfiable.

Notice that the 
 \prune\ rules play no part
 in the soundness proof.
 A ticked branch encodes a model
even if a \prune\ rule
is not applied when it could be.

\section{Proof of Completeness:}
\label{sec:complete}

We have to show that
if a formula
has a model then
it has a successful tableau.
This time we will use the model 
to find the tableau.
The basic idea 
is to use a model (of the satisfiable formula)
to show that
{\em in any tableau}
there will be a branch 
(i.e. a leaf) with a tick.

A weaker result is to show that
there is some tableau with a leaf with a tick.
Such a weaker result 
may actually be ok 
to establish correctness
and complexity of the tableau technique.
However, it 
raises questions about whether
a ``no'' answer from a tableau
is correct
and
it does not give clear
guidance for the implementer.
We show the stronger result:
it does not matter which
order static rules are applied.

\begin{lemma}[Completeness]
Suppose that $\phi$ is a satisfiable formula
of LTL.
Then any finished tableau for $\phi$ will
be successful.
\end{lemma}

\begin{proof}
Suppose that $\phi$ is a satisfiable formula
of LTL.
It will have a model. Choose one, say
$(S,R,g), \sigma \models \phi$.
In what follows we (use standard practice when the
model is fixed and) write
$\sigma_{\geq i} \models \alpha$
when we mean
$(S,R,g), \sigma_{\geq i} \models \alpha$.

Also, build a completed tableau $T$ for
$\phi$ in any manner as long as the rules are followed.
Let $\Gamma(x)$ be the formula
set label on the node $x$ in $T$.
We will show that 
$T$ has a ticked leaf.

To do this we will
construct a sequence
$x_0, x_1, x_2, ....$
of nodes, with $x_0$ being the root.
This sequence may terminate at a tick
(and then we have succeeded)
or it may hypothetically go on forever
(and more on that later).
In general, the sequence 
will head downwards
from a parent to a child node
but 
occasionally it may 
jump back up to an ancestor.

As we go we will
make sure that each
node $x_i$ is associated with
an index $J(i)$ along the fullpath $\sigma$
and we guarantee the following invariant
$INV(x_i,J(i))$ for each $i \geq 0$.
The relationship
$INV(x,j)$ is that
for each $\alpha \in \Gamma(x)$,
$ \sigma_{\geq j} \models \alpha$.

Start by putting $J(0)=0$
when $x_0$ is the tableau root node.
Note that
the only formula in
$\Gamma(x_0)$ is $\phi$ and
that 
$ \sigma_{\geq 0} \models \phi$.
Thus $INV(x_0,J(0))$ holds at the start.

Now suppose that we have identified the $x$ sequence up until
$x_{i}$.
Consider the rule that is used in $T$ to
extend a tableau branch from $x_i$ to some
children.
Note that we can also ignore the 
cases in which the
rule is EMPTY or LOOP
because they would immediately give us the
ticked branch that is sought.

It is useful to define the sequence advancement procedure
in the cases apart from the PRUNE rules separately.
Thus we now describe a procedure,
call it $A$,
that is given a node $x$ and index $j$
satisfying
$INV(x,j)$
and, in case that the node $x$ has children via 
any rule except PRUNE,
the procedure $A$ will give us
a child node $x'$ and index $j'$
which is either $j$ or $j+1$,
such that $INV(x',j')$ holds.
The idea will be to use procedure $A$
on $x_i$ and $J(i)$
to get $x_{i+1}$ and $J(i+1)$
in case the PRUNE rule is 
not used at node $x_i$.
We return to deal with advancing from $x_i$
in case that a PRUNE rule is used later.
So now we describe procedure $A$
with $INV(x,j)$ assumed.

[EMPTY]
If $\Gamma(x)= \{ \}$
then
we are done. $T$ is a successful tableau as required.

[CONTRADICTION]
Consider if it is possible for us to reach
a leaf at $x$ with a cross
because of a contradiction.
So there is some $\alpha$ with
$\alpha$ and $\neg \alpha$ in $\Gamma(x)$.
But this can not happen as then
$ \sigma_{\geq j} \models \alpha$
and
$ \sigma_{\geq j} \models \neg \alpha$.

[$\neg \neg$-rule]
So $\neg \neg \alpha$ is in $\Gamma(x)$
and there is one child,
which we will make
$x'$
and we will put
$j'=j$.
Because
$ \sigma_{\geq j} \models \neg \neg \alpha$
we also have
$ \sigma_{\geq j'} \models \alpha$.
Also 
for every other $\beta \in \Gamma(x') \subseteq
\Gamma(x) \cup \{ \alpha \}$,
we still have
$ \sigma_{\geq j'} \models \beta$.
So we have the invariant holding.

{
[$\wedge$-rule]}
So $\alpha \wedge \beta$ is in $\Gamma(x)$
and there is one child,
which we will make
$x'$
and we will put
$j'=j$.
Because
$ \sigma_{\geq j} \models \alpha \wedge \beta$
we also have
$ \sigma_{\geq j'} \models \alpha$
and
$ \sigma_{\geq j'} \models \beta$.
Also 
for every other $\gamma \in \Gamma(x') \subseteq
\Gamma(x_i) \cup \{ \alpha, \beta \}$,
we still have
$ \sigma_{\geq j'} \models \gamma$.
So we have the invariant holding.

{
[$\neg \wedge$-rule]}
So $\neg (\alpha \wedge \beta)$ is in $\Gamma(x)$
and there are two children.
One $y$ is labelled
$(\Gamma(x) \setminus \{ \neg ( \alpha \wedge \beta) \} ) \cup \{ \neg \alpha \}$ and
the other, $z$,
is labelled
$( \Gamma(x) \setminus \{ \neg ( \alpha \wedge \beta) \} ) \cup \{ \neg \beta \}$.
We know
$ \sigma_{\geq j} \models \neg ( \alpha \wedge \beta)$.
Thus,
$ \sigma_{\geq j} \not \models \alpha \wedge \beta$
and it is not the case that
both
$ \sigma_{\geq j} \models \alpha$
and
$ \sigma_{\geq j} \models \beta$.
So either
$ \sigma_{\geq j} \models \neg \alpha$
or
$ \sigma_{\geq j} \models \neg \beta$.

If the former,
i.e. that
$ \sigma_{\geq j} \models \neg \alpha$
we will make $x'=y$
and otherwise we will
make
$x'=z$.
In either case put
$j'=j$.
Let us check the invariant.
Consider the first case. The other is exactly analogous.

We already know that we have
$ \sigma_{\geq j'} \models \neg \alpha$.
Also 
for every other $\gamma \in \Gamma(x') =
\Gamma(y) \subseteq
\Gamma(x) \cup \{ \neg \alpha \}$,
we still have
$ \sigma_{\geq j'} \models \gamma$.
So we have the invariant holding.

[$U$-rule]
So $\Gamma(x)= \Delta \cupdot \{ \alpha U \beta \}$
and there are two children.
One $y$ is labelled
$\Gamma(y)= \Delta \cup \{ \beta \}$
 and
the other, $z$,
is labelled
$\Gamma(z)= \Delta \cup \{ \alpha, X(\alpha U \beta) \}$.
We know
$ \sigma_{\geq j} \models \alpha U \beta$.
Thus,
there is some $k \geq j$ such that
$\sigma_{\geq k} \models \beta$
and
for all $l$,
if
$j \leq l < k$
then
$ \sigma_{\geq l} \models \alpha$.
If 
$ \sigma_{\geq j} \models \beta$
then 
we can choose $k=j$
(even if other choices as possible)
and otherwise
choose any such $k > j$.
Again there are two cases,
either $k=j$
or $k> j$.

In the first case, when 
$ \sigma_{\geq j} \models \beta$,
we put
$x'=y$
and otherwise we will
make
$x'=z$.
In either case put
$j'=j$.

Let us check the invariant.
Consider the first case.
We have
$ \sigma_{\geq j'} \models \beta$.

In the second case,
we know that we have
$ \sigma_{\geq j'} \models \alpha$
and
$ \sigma_{\geq j'+1} \models \alpha U \beta$.
Thus
$ \sigma_{\geq j'} \models X( \alpha U \beta)$.

Also, in either case,
for every other $\gamma \in \Gamma(x')$
we still have
$ \sigma_{\geq j'} \models \gamma$.
So we have the invariant holding.

[$\neg U$-rule]
So $\Gamma(x) = \Delta \cupdot \{ \neg( \alpha U \beta)\}$
and there are two children.
One $y$ is labelled
$\Delta \cup \{ \neg \alpha, \neg \beta \}$
 and
the other, $z$,
is labelled
$\Delta \cup \{ \neg \beta, X\neg (\alpha U \beta) \}$.
We know
$ \sigma_{\geq j} \models \neg (\alpha U \beta)$.
So for sure
$ \sigma_{\geq j} \models \neg \beta$.

Furthermore,
possibly
$ \sigma_{\geq j} \models \neg \alpha$
as well,
but otherwise
if
$ \sigma_{\geq j} \models \alpha$
then we can show that we can not have
$ \sigma_{\geq j+1} \models \alpha U \beta$.
Suppose for contradiction
that
$ \sigma_{\geq j} \models \alpha$
and
$\sigma_{\geq j+1} \models \alpha U \beta$.
Then 
there is some $k \geq j$ such that
$ \sigma_{\geq k} \models \beta$
and
for all $l$,
if
$j \leq l < k$
then
$ \sigma_{\geq l} \models \alpha$.
Thus
$\sigma_{\geq j} \models \alpha U \beta$.
Contradiction.

So we can conclude that there
are two cases when the $\neg U$-rule is used.
CASE 1:
$ \sigma_{\geq j} \models \neg \beta$
and
$ \sigma_{\geq j} \models \neg \alpha$.
CASE 2:
$ \sigma_{\geq j} \models \neg \beta$
and
$ \sigma_{\geq j+1} \models \neg (\alpha U \beta)$.

In the first case, when 
$\sigma_{\geq j} \models \neg \beta$,
we put
$x'=y$
and otherwise we will
make
$x'=z$.
In either case put
$j'=j$.

Let us check the invariant.
In both cases
we know that we have
$ \sigma_{\geq j'} \models \neg \beta$.
Now consider the first case.
We also have 
$ \sigma_{\geq j} \models \neg \alpha$.
In the second case,
we know that we have
$ \sigma_{\geq j+1} \models \neg (\alpha U \beta)$.
Thus
$ \sigma_{\geq j'} \models X\neg (\alpha U \beta)$.
Also, in either case,
for every other $\gamma \in \Gamma(x')$
we still have
$ \sigma_{\geq j'} \models \gamma$.
So we have the invariant holding.

[OTHER STATIC RULES]: similar.

{
[TRANSITION]}
So $\Gamma(x)$ is poised
and there is one child,
which we will make
$x'$
and we will put
$j'=j+1$.

Consider a formula
 $\gamma \in \Gamma(x')=
\{ \alpha |
X \alpha \in \Gamma(x) \}
\cup
\{ \neg \alpha |
\neg X \alpha \in \Gamma(x) \}$.

CASE 1: Say that
$X \gamma \in \Gamma(x)$.
Thus, by the invariant,
$\sigma_{\geq j} \models X \gamma$.
Hence,
$\sigma_{\geq j+1} \models \gamma$.
But this is just
$\sigma_{\geq j'} \models \gamma$
as required.

CASE 2: Say that $\gamma= \neg \delta$
and
$\neg X \delta \in \Gamma(x)$.
Thus, by the invariant,
$\sigma_{\geq j} \models \neg X \delta$.
Hence,
$\sigma_{\geq j+1} \not \models \delta$.
But this is just
$\sigma_{\geq j(i+1)} \models \gamma$
as required.

So we have the invariant holding.

{
[LOOP]}
If, in $T$, the node $x_i$
is a leaf just getting a tick via the 
LOOP rule
then
we are done. 
$T$ is a successful tableau as required.

So that ends the description of procedure $A$
that is given a node $x$ and index $j$
satisfying
$INV(x,j)$
and, in case that the node $x$ has children via 
any rule except PRUNE or \prunez,
the procedure $A$ will give us
a child node $x'$ and index $j'$,
which is either $j$ or $j+1$,
such that $INV(x',j')$ holds.
We use procedure $A$ to
construct a sequence
$x_0, x_1, x_2, ....$
of nodes, with $x_0$ being the root.
and guarantee the invariant
$INV(x_i,J(i))$ for each $i \geq 0$.

The idea will be to use procedure $A$
on $x_i$ and $J(i)$
to get $x_{i+1}$ and $J(i+1)$
in case the PRUNE rule is 
not used at node $x_i$.
Start by putting $J(0)=0$
when $x_0$ is the tableau root node.
We have seen that $INV(x_0,J(0))$ holds at the start.

{[\prune\ ]}
Now, we complete the description of the construction
of the $x_i$ sequence by 
explaining what to do 
in case
$x_i$ is a node on which PRUNE
is used.
Suppose that
$x_i$
is a node which gets a cross
in $T$
via the \prune\ rule.
So there is a sequence
$u=x_h, x_{h+1}, ...,
x_{h+a}=v, x_{h+a+1}, ...,
x_{h+a+b}=x_i=w$
such that
$\Gamma(u)=\Gamma(v)=\Gamma(w)$
and no extra eventualities of
$\Gamma(u)$ are satisfied between
$v$ and $w$ that were
not already satisfied
between $u$ and $v$.

What we do now is to 
undertake a sort of backtracking exercise
in our proof.
We
choose some such
$u$, $v$ and $w$,
there may be more than one triple,
and
proceed with the construction
as if 
$x_i$ was $v$ instead of $w$.
That is we use
the procedure
$A$
on $v$ with $J(i)$
 to get from $v$ to
 one $x_{i+1}$ of its children
 and define $J(i+1)$.
Procedure $A$ above 
can be applied because
$\Gamma(v)=\Gamma(x_i)$
and so the invariant holds
for $v$ with $J(i)$ as well as for  $x_i$ with $J(i)$.

Thus we keep going
with the new $x_{i+1}$ child of $v$,
and $J(i)$.

If the variant \prunez\ rule is
used on $x_i$ then
the action is similar but simpler.
So there is a sequence
$u=x_h, x_{h+1}, ...,
x_{h+a}=x_i=v$
such that
$\Gamma(u)=\Gamma(v)$
and no eventualities of
$\Gamma(u)$ are satisfied between
$u$ and $v$ but there is at least
one eventuality in $\Gamma(u)$.

What we do now is to 
choose some such
$u$, and $v$,
and
proceed with the construction
as if 
$x_i$ was $u$ instead of $v$.
That is we use
the procedure
$A$
on $u$ with $J(i)$
 to get from $u$ to
 one $x_{i+1}$ of its children
 and define $J(i+1)$.
Procedure $A$ above 
can be applied because
$\Gamma(u)=\Gamma(x_i)$
and so the invariant holds
for $u$ with $J(i)$ as well as for  $x_i$ with $J(i)$.
Thus we keep going
with the new $x_{i+1}$ child of $u$,
and $J(i)$.

Now let us consider 
whether the above construction
goes on for ever.
Clearly it may
end finitely with
us 
finding a ticked leaf and
succeeding.
However,
at least in theory,
it may seem possible that the
construction keeps going forever
even though the tableau will be 
finite.
The rest of the proof
is to show that this 
actually can not happen.
The construction can not go on forever.
It must stop and
the only way that we have 
shown that that can happen is
by finding a tick.

Suppose for contradiction that the construction does go on forever.
Thus, because there are only a finite number of nodes in the tableau,
we must meet the \prune\ 
{(or \prunez)}
rule and
jump back up the tableau infinitely often.

\newcommand{\triple}{{tuple}}

When we do find an application of the \prune\ rule with
triple $(u,v,w)$ of nodes from $T$ call that a
jump triple.
Similarly, 
when we find an application of the \prunez\ rule with
pair $(u,v)$ of nodes from $T$ call that a
jump pair.
Let jump {\em tuples} be either
 jump pairs or jump triples.

There are only a finite number of jump \triple s
so there must be some that
cause us to jump infinitely often.
Call these {\em recurring}
jump \triple s.

Say that
$(u_0,v_0, w_0)$ {or $(u_0,v_0)$}
is
one such.
We can choose $u_0$ so that 
for no other recurring jump triple
$(u_1,  v_1, w_1)$
{or pair $(u_1,v_1)$}
do we have
$u_1$ being a proper ancestor of $u_0$.

As we proceed through the construction
of $x_0, x_1, ..$
and see a jump every so often,
eventually all the 
jump \triple s
who only cause a jump a finite number
of times
stop causing jumps.
After that time,
$(u_0,v_0, w_0)$
{or $(u_0,v_0)$}
will still cause a jump
every so often.

Thus after that time
$u_0$ will never appear
again as  the $x_i$ that we choose
and all the $x_i$s that we choose will
be descendants of $u_0$.
This is because we will never jump
up to $u_0$ or above it (closer to the root).
Say that $x_N$ is the very last $x_i$ that
is equal to $u_0$.

Now consider any $X(\alpha U \beta)$ that appears
in $\Gamma(u_0)$.
(There must be at least one eventuality
in $\Gamma(u_0)$ as it is used to
apply rule \prune
{ or \prunez}).

A simple induction shows that
$\alpha U \beta$
or $X( \alpha U \beta)$
will appear in every
$\Gamma(x_i)$
from $i=N$
up until at least when
$\beta$ appears
in some $\Gamma(x_i)$ after that
(if that ever happens).
This is because
if $\alpha U \beta$
is in $\Gamma(x_i)$
and $\beta$ is not there
and does not get put there
then
$X(\alpha U \beta)$ will also be put
in before the next temporal TRANSITION rule.
Each temporal TRANSITION rule
will thus put $\alpha U \beta$
into the new label.
Finally, in case the $x_i$ sequence meets a PRUNE jump
$(u,v,w)$
then the new $x_{i+1}$ will be a child 
of $v$ which is a descendent of $u$
which is a descendent of $u_0$
so will also contain 
$\alpha U \beta$
or $X(\alpha U \beta)$.
{Similarly with \prunez\ jumps.}

Now $J(i)$ just
increases by $0$ or $1$ 
with each increment of $i$,
We also know that
$\sigma_{\geq J(i)} \models \alpha U \beta$
from $i=N$ onwards
until (and if) $\beta$ gets put in $\Gamma(x_i)$.
Since $\sigma$ is a fullpath
we will eventually get to some
$i$
with
$\sigma_{\geq J(i)} \models \beta$.
In that case our
construction makes us put
$\beta$ in the label.
Thus we do eventually get to some
$i \geq N$
with $\beta \in \Gamma(x_i)$.
Let $N_\beta$ be the first
such $i \geq N$.
Note that
all the nodes between
$u_0$ and
$x_{N_\beta}$
in the tableau
also appear as
$x_i$ for 
$N < i < N_\beta$
so that
they all have
$\alpha U \beta$ and not $\beta$
in their labels
$\Gamma(x_i)$.

Now let us consider if 
we ever
jump up above 
$x_{N_\beta}$
at any TRANSITION of our construction
(after $i=N_\beta$).
In that case 
there would be a PRUNE jump triple
of tableau nodes
$u$, $v$ and $w$
governing the first such jump
{or possibly a \prunez\ jump.
Consider first a \prune\ jump}.
Since $u$ is not above
$u_0$
and $v$ is above 
$x_{N_\beta}$,
we must have 
$\Gamma(u) = \Gamma(v)$
with 
$X(\alpha U \beta)$
in them
and $\beta$ not satisfied in between.
But $w$ will be below
$x_{N_\beta}$
at the first such jump,
meaning that
$\beta$ is satisfied
between
$v$ and $w$.
That is
a contradiction to the PRUNE rule being applicable to this triple.

Now consider a \prunez\ jump
from $v$ below $x_{N_\beta}$
up to $u$ above it.
Since $u$ is not above
$u_0$
and $v$ is below
$x_{N_\beta}$,
we must have 
$\Gamma(u) = \Gamma(v)$
with 
$X(\alpha U \beta)$
in them
and not satisfied in between.
But $v$ will be below
$x_{N_\beta}$
at the first such jump,
meaning that
$\beta$ is satisfied
between
$u$ and $v$.
That is
a contradiction to the \prunez\ rule being applicable to this triple.

Thus the $x_i$ sequence 
stays within descendants of $x_{N_\beta}$
forever after $N_\beta$.

The above reasoning applies to all
eventualities in $\Gamma(u_0)$.
Thus, after they are each satisfied,
the construction $x_i$
does not jump up above any of them.
When the next supposed
jump involving
$u_0$ with some $v$ and
{ (perhaps)}
 $w$
happens after that
it is clear that
all of the eventualities
in $\Gamma(u_0)$
are satisfied above $v$.

This is a contradiction
to such a jump ever happening.
Thus we can conclude that
there are not an infinite number of jumps after all.
The construction must finish with a tick.
This is the end of the completeness proof.
\end{proof}

\section{Complexity and Implementation}
\label{sec:complex}

Deciding LTL satisfiability is in PSPACE \cite{SiC85}.
(In fact our tableau approach can be used to show that via
Savitch's theorem by 
assuming guessing the right branch.)

The tableau search through the new tableau,
even in a non-parallel implementation,
should (theoretically) be able to be implemented to run
faster than that through
the state of the art tableau technique
of \cite{Schwe98}.
This is because
there is less information to keep track of and
no backtracking from potentially successful branches
when a repeated label is discovered.

A variety of implementations are currently underway
at Udine, including some comparative
experiments with
other available LTLSAT checkers.
Early results are very promising and 
publications on the results
will be forthcoming.
For now, 
as this paper is primarily about the
theory behind the new rules,
we have provided
a demonstration Java implementation
to allow readers to experiment with
the way that the tableau works.
The program allows comparison with 
a corresponding implementation of
the Schwendimann
tableau.
The demonstration Java implementation
is available at
\webpage.
This allows users to understand the tableau building process
in a step by step way.
It is not 
designed as a fast implementation.
However, it does report on
ow many tableau construction steps were taken.

Detailed comparisons of the running times
are available 
\shortversion{in \cite{ltlsattabLONG}}
\longversion{via the web page.}
In Figure~\ref{fig:compare},
we give a small selection to give the
idea of the experimental results.
This is just on a quite long formula,
``Rozier 9",
one very long formula
``anzu amba amba 6"
from the so-called Rozier counter example series
of \cite{VSchuppanLDarmawan-ATVA-2011}
and an interesting property
``foo4" (described below).
Shown is formula length,
running time in seconds (on a standard laptop),
number of tableau steps and
the maximum depth of a branch in
poised states.
As claimed, on many examples
the new tableau needs roughly the same number
of steps and the same amount of time
on each step.
However, there are interesting formulas
(such as the foo series)
for which the new tableau makes
a significant saving.
 
\begin{figure}
\begin{center}
$\;$\\
\begin{tabular}{|c|c|c|c|c|c|c|c|}
\hline
\multicolumn{2}{|c|}{fmla} &  \multicolumn{3}{c|}{Reynolds} & \multicolumn{3}{c|}{Schwendimann}\\
\hline
& length & sec & steps & depth & sec & steps & depth \\
\hline
r9 & 277 & 109 & 240k & 4609 & 112 & 242k & 4609 \\
as6 & 1864 & 0.001 & 54 & 2 & 0.001 & 55 & 2 \\
foo4 & 84 & 7.25 & 7007k & 9 & 15.9 & 16232k & 3 \\
\hline
\end{tabular}
\end{center}
\caption{Comparison of the two tableaux from the Java implementation}
\label{fig:compare}
\end{figure}

We use benchmark
formulas from the
various series
described in \cite{VSchuppanLDarmawan-ATVA-2011}.
The only new series
we add
is $\mbox{foo}_n$
as follows:
for all $n \geq 2$,
\[
\mbox{foo}_n=
a
\wedge 
G(a \leftrightarrow X \neg a)
\wedge \bigwedge_{i=1}^n GF b_i
\wedge \bigwedge_{i=1}^n G( b_i  \rightarrow  \neg a)
\wedge \bigwedge_{i=1}^{n-1} \bigwedge_{j=i+1}^n G \neg( b_i \wedge b_j )
\]
Originally, the $\mbox{foo}_n$ series
was invented by the author to
obtain an idea experimentally
how much longer it might take the
new tableau compared to
Schwendimann's
on testing examples
but the experiments show
the opposite outcome.

\section{Comparisons with the Schwendimann Tableau}
\label{sec:compare}

In this section we
give a detailed comparison of the
new tableau's operation
compared to that of the
Schwendimann's tableau.

\subsection{Schwendimann's Tableau}

From \cite{Schwe98}.

The tableau is a labelled tree. Suppose we are to decide the
satisfiability or not of $\phi$.

Labels are of the form
$(\Gamma, S, R)$ where
$\Gamma$ is a set of formulas
(subformulas of $\phi$),
$S=(Ev,Br)$ is a pair
(described shortly)
and
$R$ is a pair
of the form
$(n,uev)$ where $n \in \natn$
and $uev \subseteq \clos(\phi)$.

The $\Gamma$ part of a node label
$(\Gamma, S, R)$
is a set of subformulas
of $\phi$
serving a similar purpose 
to our node labels.

The second component
$S=(Ev,Br)$ is a pair
where $Ev$ is a set of formulas,
and $Br$ is a
list of pairs
each of the form $(\Gamma_i,E_i)$.
The set $Ev$
records the eventualities cured at the
current state.
The list
$Br=\langle (\Gamma_1,Ev_1), ..., (\Gamma_m,Ev_m) \rangle$
is a 
record of useful
information from the
state ancestors of the current node
along the current branch.
Each $(\Gamma_i,Ev_i)$
records the poised label $\Gamma_i$ of the $i$th state
down the branch
and the set $Ev_i$ of eventualities fulfilled there.

The second component $S$
essentially contains information about the
labels on the current branch of the tree
and so is just a different way of managing
the same sort of ``historical'' recording
that we manage in our new tableau.
We manage the historical records
by having the branch ancestor labels directly
available for checking.

The third component
$(n,uev)$,  a pair consisting
of a number and a set of formulas,
is specific to the Schwendimann
tableau process
and has no analogous component
in our new tableau.
The number $n$ records the highest 
index of ancestor state 
in the current branch that
can be reached directly
by an up-link
from a descendent node
of the current node.
The set $uev$ contains the
eventualities which are
in the current node label
but which are not cured
by the time of the 
end of the branch below.

The elements of the third component
are not known 
as the tableau is constructed 
until 
all descendants
on all branches
below the current node
are expanded.
Then the values can be filled
in using some straightforward 
but slightly lengthy rules
about how to compute them
from children to parent nodes.

The rules for working out these
components
are slightly complicated in the
case of
disjunctive rules,
including expansion of $U$ formulas.
The third component
helps assess when a 
branch and the whole tableau
can be finished.

Apart from having to deal with the
second and third components
of the labels,
the rules are largely similar
to the rules for the new tableau.
However,
there are no prune rules.
The Schwendimann tableau
does not continue a branch
if there is a node with the
same 
label as one of its proper ancestors.
The branch stops there.

The other main difference to note is that
the Schwendimann tableau construction
rules
for disjunctive formulas
do in general need to
combine information from both
children's branches
to compute the
label on the parent.

\subsection{Example}

Consider the example
\[
\theta= \mbox{foo}_2=
a
\wedge 
G(a \leftrightarrow X \neg a)
\wedge GF b_1 
\wedge GF b_2 
\wedge G( b_1  \rightarrow  \neg a)
\wedge G( b_2    \rightarrow   \neg a)
\wedge G \neg( b_1 \wedge b_2 ).
\]

Technically Schwendimann tableau assumes all formulas
(the input formula and its subformulas)
are in {\em negation normal form}
which involves some rewriting so that negations
only appear before atomic propositions.
However, we can 
make a minor modification
and assume the static rules
are the same as our new ones.

We start a tableau with 
the label
$(\Gamma,S,R)$
with
$\Gamma= \{ \theta \}$,
$S=( \{ \}, \langle \rangle)$
and
$R= (n,uev)$ with both $n$ and $uev$ unknown as yet.

Tableau rules decompose
$\theta$ and subsequent subformulas
in a similar way to in our new tableau.
Neither $S$ nor $R$ change
while we do not apply a step rule (known as $X$ rule).

The conjunctions and $G$ rules are just as in our new tableau.
Thus the tableau construction
soon reaches a node with the label
$(\Gamma,S,R)$ as follows:
$S=( \{ \}, \langle \rangle)$,
$R= (n,uev)$ with both $n$ and $uev$ unknown as yet
and
\[
\Gamma = \{
a, 
a \leftrightarrow X \neg a,
XG(a \leftrightarrow X \neg a),
G(b_1 \rightarrow \neg a),
G(b_2 \rightarrow \neg a),
G\neg (b_1 \wedge b_2),
GFb_1,
GFb_2
\}.
\]

There are six choices causing branches within this state caused by the
disjuncts and the $F \beta$ formulas.
This could lead to 64 different branches before we complete the
state but many of these choices lead immediately to contradictions.
The details depend on the order of choice of decomposing formulas.

A typical expansion leads to a branch with a node labelled by the
state $(\Gamma,S,R)$ as follows:
$S=( \{ \}, \langle \rangle)$,
$R= (n,uev)$ with both $n$ and $uev$ unknown as yet
and

\[
\begin{array}{ll}
\Gamma_0 = \{ &
a, 
X \neg a,
XG(\neg a \leftrightarrow X \neg a),
\neg b_1,
XG(b_1 \rightarrow \neg a),
\neg b_2,\\
& 
XG(b_2 \rightarrow \neg a),
XG\neg (b_1 \wedge b_2),
XFb_1,
XGFb_1,
XFb_2,
XGFb_2
\}.
\end{array}
\]

Notice that the first component of $S$
is still empty as,
in this case,
none of the eventualities,
$Fb_1$ nor $Fb_2$,
is fulfilled here.

The next rule to use in such a situation
is the transition rule, which is very similar to that
in our new tableau,
except that
the second component part of the
label is updated as well.
We find
$S=( \{ \}, \langle (\{\}, \Gamma_0) \rangle)$,
$R= (n,uev)$ with both $n$ and $uev$ unknown as yet
and
\[
\Gamma = \{
\neg a,
G(\neg a \leftrightarrow X \neg a),
G(b_1 \rightarrow \neg a),
G(b_2 \rightarrow \neg a),
G\neg (b_1 \wedge b_2),
Fb_1,
GFb_1,
Fb_2,
GFb_2
\}.
\]

A further series of expansions and choices leads us to
the following new state,
$S=( \{ b_1 \}, \langle (\{\}, \Gamma_0) \rangle)$,
$R= (n,uev)$ with both $n$ and $uev$ unknown as yet
and
\[
\begin{array}{ll}
\Gamma_1 = \{ &
\neg a, 
X a,
XG(\neg a \leftrightarrow X \neg a),
b_1,
XG(b_1 \rightarrow \neg a),
\neg b_2,\\
& 
XG(b_2 \rightarrow \neg a),
XG\neg (b_1 \wedge b_2),
XGFb_1,
XFb_2,
XGFb_2
\}.
\end{array}
\]
Applying the step rule here gives us the following:
$S=( \{ \}, \langle (\{\}, \Gamma_0), (\{b_1\}, \Gamma_1) \rangle)$,
$R= (n,uev)$ with both $n$ and $uev$ unknown as yet
and
\[
\Gamma = \{ 
a,
G(\neg a \leftrightarrow X \neg a),
G(b_1 \rightarrow \neg a),
G(b_2 \rightarrow \neg a),
G\neg (b_1 \wedge b_2),
Fb_1,
GFb_1,
Fb_2,
GFb_2
\}.
\]

When this is branch is expanded further then we find
ourselves at
$(\Gamma_0, S, R)$ with
$S=( \{ \}, \langle (\{\}, \Gamma_0), (\{b_1\}, \Gamma_1) \rangle)$ and
$R= (n,uev)$ with both $n$ and $uev$ unknown as yet.
Notice the main first part of the label has
ended up being $\Gamma_0$ again.
Thus we have a situation for the LOOP rule.

The loop rule uses the facts that the state repeated has the index
$1$ in the branch above and that the eventuality $Fb_2$ has not been fulfilled in this
branch.
Thus we can fill in 
$n=1$ and $uev=\{ Fb_2 \}$ 
at all the intervening pre-states
which have been left unknown so far.
These values do not as yet transfer
up to the top state as yet as there are still
other undeveloped branches from within that state.

The other two states that we find are a 
minor variation on $\Gamma_1$,
\[
\begin{array}{ll}
\Gamma_2 = \{ &
\neg a, 
X a,
XG(\neg a \leftrightarrow X \neg a),
b_1,
XG(b_1 \rightarrow \neg a),
\neg b_2,\\
& 
XG(b_2 \rightarrow \neg a),
XG\neg (b_1 \wedge b_2),
XFb_1,
XGFb_1,
XFb_2,
XGFb_2
\},
\end{array}
\]
and a mirror image of $\Gamma_1$
when $b_2$ is true
instead of $b_1$:
\[
\begin{array}{ll}
\Gamma_3 = \{ &
\neg a, 
X a,
XG(\neg a \leftrightarrow X \neg a),
\neg b_1,
XG(b_1 \rightarrow \neg a),
b_2,\\
& 
XG(b_2 \rightarrow \neg a),
XG\neg (b_1 \wedge b_2),
XFb_1,
XGFb_1,
XFb_2,
XGFb_2
\}.
\end{array}
\]
Then all eventualities are fulfilled
and the tableau succeeds
after 3933 steps.
The overall picture of
states (poised labels) is
as follows
(with $1$ standing for $\Gamma_1$ etc).

\[
\xy
(40,20)*{0};
{\ar(40,18)*{};(20,13)*{}};
{\ar(40,18)*{};(60,13)*{}};
(20,10)*{1};
{\ar(20,7)*{};(20,3)*{}};
(20,0)*{0};
{\ar@/^1.0pc/(22,0)*{};(38,18)*{}};
{\ar(40,17)*{};(40,13)*{}};
(40,10)*{2};
{\ar(40,7)*{};(40,3)*{}};
(40,0)*{0};
{\ar@/^2.0pc/(38,0)*{};(38,18)*{}};
(60,10)*{3};
{\ar(60,7)*{};(60,3)*{}};
(60,0)*{0};
{\ar@/_1.0pc/(58,0)*{};(42,18)*{}};
(60,-5)*{\surd};
\endxy
\]

\subsection{Same Example in New Tableau}

By comparison, the new tableau
visits the same states 
but takes only 3087 steps to
proceed as follows:

\[
\xy
(40,40)*{0};
{\ar(40,37)*{};(40,33)*{}};
(40,30)*{1};
{\ar(40,27)*{};(40,23)*{}};
(40,20)*{0};
{\ar(40,18)*{};(20,13)*{}};
{\ar(40,18)*{};(60,13)*{}};
(20,10)*{1};
(20,5)*{\times};
{\ar(40,17)*{};(40,13)*{}};
(40,10)*{2};
{\ar(40,7)*{};(40,2)*{}};
(40,0)*{0};
(40,-5)*{\times};
(60,10)*{3};
{\ar(60,7)*{};(60,2)*{}};
(60,0)*{0};
{\ar@/_2.0pc/(62,0)*{};(42,40)*{}};
(60,-5)*{\surd};
\endxy
\]

\subsection{Comparisons}

The new tableau can always decide on the basis of a single branch, working downwards. Schwendimann's needs communication up and between branches, and
will generally require full development of several branches.

The new tableau just needs to store formula labels down current branch. 
However, if doing depth-first search,
also needs information about choices made to enable backtracking.
Schwendimann's tableau needs to pass extra sets of unfulfilled eventualities 
back up branches
also keeping track of indices while backtracking.
It also needs to store 
 information about choices made to enable backtracking.

\section{Conclusion}
\label{sec:concl}

We have introduced 
novel tableau construction rules
which support a new
tree-shaped, one-pass tableau system
for LTLSAT.
It is traditional in style, simple in all aspects
with no extra notations on nodes,
neat to introduce to students,
amenable to manual use and
promises efficient and fast automation.

In searching or constructing the tableau
one can explore down branches completely independently
and further break up the search down individual
branches into separate
 independent
processes.
Thus it is particularly suited to parallel implementations.

Experiments show that, even in a standard depth-first search implementation,
it is also
competitive with the current state of the art
in tableau-based approaches to
LTL satisfiability checking.
This is good reason to believe that
very efficient implementations can
be achieved as the
step by step construction task
is particularly simple and
requires on minimal storage and
testing.

Because of the simplicity,
it also seems to be a good base for more intelligent
and sophisticated algorithms:
including heuristics for choosing amongst
branches and ways of managing sequences
of label sets.

The idea of the PRUNE rules
potentially have many other applications.


\begin{thebibliography}{KMMP93}

\bibitem[BA12]{Mor12}
Mordechai Ben-Ari.
\newblock Propositional logic: Formulas, models, tableaux.
\newblock In {\em Mathematical Logic for Computer Science}, pages 7--47.
  Springer London, 2012.

\bibitem[Bet55]{Beth55}
E.~Beth.
\newblock Semantic entailment and formal derivability.
\newblock {\em Mededelingen der Koninklijke Nederlandse Akad. van Wetensch},
  18, 1955.

\bibitem[CGH97]{CGH97}
E.~M. Clarke, O.~Grumberg, and K.~Hamaguchi.
\newblock Another look at ltl model checking.
\newblock {\em Formal Methods in System Design}, 10(1):47--71, 1997.

\bibitem[FDP01]{FDP01}
Michael Fisher, Clare Dixon, and Martin Peim.
\newblock Clausal temporal resolution.
\newblock {\em ACM Transactions on Computational Logic (TOCL)}, 2(1):12--56,
  2001.

\bibitem[Fos95]{Fos95}
Ian Foster.
\newblock {\em Designing and Building Parallel Programs}.
\newblock Addison Wesley, 1995.

\bibitem[Gir00]{Girle00}
Rod Girle.
\newblock {\em Modal Logics and Philosophy}.
\newblock Acumen, Teddington, UK, 2000.

\bibitem[GKS10]{Goranko2010113}
Valentin Goranko, Angelo Kyrilov, and Dmitry Shkatov.
\newblock Tableau tool for testing satisfiability in ltl: Implementation and
  experimental analysis.
\newblock {\em Electronic Notes in Theoretical Computer Science}, 262(0):113 --
  125, 2010.
\newblock Proceedings of the 6th Workshop on Methods for Modalities (M4M-6
  2009).

\bibitem[Gor10]{PLTL}
Raj Gore.
\newblock pltl tableau implementation, 2010.

\bibitem[Gou84]{Gou84}
G.~D. Gough.
\newblock Decision procedures for temporal logic.
\newblock Master's thesis, Department of Computer Science, University of
  Manchester, 1984.

\bibitem[Gou89]{Gou89}
G.~Gough.
\newblock Decision procedures for temporal logics.
\newblock Technical Report UMCS-89-10-1, Department of Computer Science,
  University of Manchester, 1989.

\bibitem[HK03]{hustadt2003trp++}
Ullrich Hustadt and Boris Konev.
\newblock Trp++ 2.0: A temporal resolution prover.
\newblock In {\em Automated Deduction--CADE-19}, pages 274--278. Springer
  Berlin Heidelberg, 2003.

\bibitem[Hol97]{Hol97}
G.J. Holzmann.
\newblock The model checker {SPIN}.
\newblock {\em IEEE Trans. on Software Engineering, Special issue on Formal
  Methods in Software Practice}, 23(5):279--295, 1997.

\bibitem[KMMP93]{DBLP:conf/cav/KestenMMP93}
Yonit Kesten, Zohar Manna, Hugh McGuire, and Amir Pnueli.
\newblock A decision algorithm for full propositional temporal logic.
\newblock In Costas Courcoubetis, editor, {\em CAV}, volume 697 of {\em Lecture
  Notes in Computer Science}, pages 97--109. Springer, 1993.

\bibitem[LH10]{DBLP:journals/aicom/LudwigH10}
Michel Ludwig and Ullrich Hustadt.
\newblock Implementing a fair monodic temporal logic prover.
\newblock {\em AI Commun.}, 23(2-3):69--96, 2010.

\bibitem[LP85]{LiP85}
O.~Lichtenstein and A.~Pnueli.
\newblock Checking that finite state concurrent programs satisfy their linear
  specification.
\newblock In {\em Proc. 12th ACM Symp. on Princ. Prog. Lang.}, 1985.

\bibitem[LP00]{LP00:igpl}
O.~Lichtenstein and A.~Pnueli.
\newblock Propositional temporal logic: Decidability and completeness.
\newblock {\em IGPL}, 8(1):55--85, 2000.

\bibitem[LWB10]{LWB}
LWB.
\newblock Lwb tableau implementation, 2010.

\bibitem[LZP{\etalchar{+}}13]{LZPVH13}
J.~Li, L.~Zhang, G.~Pu, M.~Vardi, and J.~He.
\newblock Ltl satisfiability checking revisited.
\newblock In {\em TIME 13}, 2013.

\bibitem[MP95]{MaP95}
Zohar Manna and Amir Pnueli.
\newblock {\em Temporal verification of reactive systems: safety}.
\newblock Springer-Verlag New York, Inc., New York, NY, USA, 1995.

\bibitem[Pnu77]{Pnu77}
A.~Pnueli.
\newblock The temporal logic of programs.
\newblock In {\em Proceedings of the Eighteenth Symposium on Foundations of
  Computer Science}, pages 46--57, 1977.
\newblock Providence, RI.

\bibitem[Rey11]{Rey:startab}
Mark Reynolds.
\newblock A tableau-based decision procedure for {CTL*}.
\newblock {\em Journal of Formal Aspects of Computing}, pages 1--41, August
  2011.

\bibitem[RV07]{DBLP:conf/spin/RozierV07}
Kristin~Y. Rozier and Moshe~Y. Vardi.
\newblock Ltl satisfiability checking.
\newblock In Dragan Bosnacki and Stefan Edelkamp, editors, {\em SPIN}, volume
  4595 of {\em Lecture Notes in Computer Science}, pages 149--167. Springer,
  2007.

\bibitem[RV11]{RV11}
K.~Y. Rozier and M.~Y. Vardi.
\newblock A multi-encoding approach for ltl symbolic satisfiability checking.
\newblock In {\em 17th International Symposium on Formal Methods (FM2011)},
  volume 6664 of {\em LNCS}, pages 417--431. Springer Verlag, 2011.

\bibitem[SC85]{SiC85}
A.~Sistla and E.~Clarke.
\newblock Complexity of propositional linear temporal logics.
\newblock {\em J. ACM}, 32:733--749, 1985.

\bibitem[Sch98a]{Schwe98}
S.~Schwendimann.
\newblock A new one-pass tableau calculus for {PLTL}.
\newblock In Harrie C.~M. de~Swart, editor, {\em Proceedings of International
  Conference, TABLEAUX 1998, Oisterwijk}, LNAI 1397, pages 277--291. Springer,
  1998.

\bibitem[Sch98b]{Sch98}
Stefan Schwendimann.
\newblock {\em Aspects of Computational Logic}.
\newblock {PhD}, Institut f{\"u}r Informatik und angewandte Mathematik, 1998.

\bibitem[SD11]{VSchuppanLDarmawan-ATVA-2011}
Viktor Schuppan and Luthfi Darmawan.
\newblock Evaluating {LTL} satisfiability solvers.
\newblock In Tevfik Bultan and Pao-Ann Hsiung, editors, {\em ATVA'11}, volume
  6996 of {\em Lecture Notes in Computer Science}, pages 397--413. Springer,
  2011.

\bibitem[SGL97]{SGL97}
P.~Schmitt and J.~Goubault-Larrecq.
\newblock A tableau system for linear-time temporal logic.
\newblock In {\em TACAS 1997}, pages 130--144, 1997.

\bibitem[Smu68]{Smu68}
R.~Smullyan.
\newblock {\em First-order Logic}.
\newblock Springer, 1968.

\bibitem[VW94]{VaW94}
M.~Vardi and P.~Wolper.
\newblock Reasoning about infinite computations.
\newblock {\em Information and Computation}, 115:1--37, 1994.

\bibitem[Wol83]{Wol83}
P.~Wolper.
\newblock Temporal logic can be more expressive.
\newblock {\em Information and computation}, 56(1--2):72--99, 1983.

\bibitem[Wol85]{Wol85}
P.~Wolper.
\newblock The tableau method for temporal logic: an overview.
\newblock {\em Logique et Analyse}, 28:110--111, June--Sept 1985.

\end{thebibliography}

\newcommand{\etalchar}[1]{$^{#1}$}

\end{document}